\documentclass[10pt]{article}
\usepackage[utf8]{inputenc}
\usepackage{amsthm}
\usepackage{amsmath}
\usepackage{amssymb}
\usepackage{tikz}
\usepackage{algorithm}
\usepackage{algpseudocode}
\usepackage{booktabs}
\usepackage{mathtools}
\usepackage{threeparttable}
\usepackage{array,tabularx,ragged2e}

\usetikzlibrary{arrows.meta,calc,decorations.markings,decorations.pathmorphing}
\newcommand{\niceArrow}{-{Stealth[length=2.25mm]}}
\theoremstyle{plain}

\usepackage[mathlines,pagewise]{lineno}
\pgfdeclarelayer{background}
\pgfdeclarelayer{foreground}
\pgfsetlayers{background,main,foreground}
\tikzset{
    my box/.style = {
        , line cap = round
        , line join = round
    }
  }
\newcommand{\highlight}[3]{
  \path [my box, line width = #1, draw = #2, transparency group, opacity=1] #3;
}
\usepackage[margin=1cm]{caption}
\usepackage{subcaption}
\usepackage[textsize=small]{todonotes}

\newtheorem{proposition}{Proposition}
\newtheorem{corollary}{Corollary}

\newtheorem{lemma}{Lemma}
\newtheorem{theorem}{Theorem}

\theoremstyle{definition}
\newtheorem{definition}{Definition}
\newtheorem{example}{Example}
\newtheorem{observation}{Observation}
\newtheorem{note}{Note}

\numberwithin{proposition}{subsection}
\numberwithin{corollary}{subsection}
\numberwithin{claim}{subsection}
\numberwithin{lemma}{subsection}
\numberwithin{theorem}{subsection}
\numberwithin{definition}{subsection}
\numberwithin{example}{subsection}
\numberwithin{observation}{subsection}
\numberwithin{note}{subsection}
\numberwithin{remark}{subsection}

\newcommand{\AND}{\textbf{and }}

\usepackage[a4paper, total={6.8in, 9.9in}]{geometry}

\usepackage{enumitem}
\setlist[itemize]{leftmargin=*}

\newcommand{\Input}{\item[\textbf{Input:}]}
\newcommand{\Output}{\item[\textbf{Output:}]}



\usepackage{xparse}

\newcounter{subroutine}
\renewcommand{\thesubroutine}{\arabic{subroutine}}

\newlength{\subrindent}
\newcommand{\callsc}[1]{\textsc{#1}}
\NewDocumentEnvironment{spec}{m m m m m m}
{%
  \refstepcounter{subroutine}
  \par\medskip
  \noindent
  \textbf{Subroutine~\thesubroutine. }\callsc{#1(#2)}
  \label{#3}
  \vspace{0.4\baselineskip}
  \begin{list}{}{\leftmargin=\subrindent \rightmargin=0pt
    \labelwidth=0pt \labelsep=0pt \itemindent=-0.4em
    \topsep=0pt
    \itemsep=4pt \parsep=0pt}
    \item \textbf{Input.}~#4
    \item \textbf{Output.}~#5
    \item \textbf{Implementation.}~#6
}
{%
  \end{list}
  \par\medskip
}

\tikzset{paralleledge/.style={to path={
      \pgfextra{%
        \pgfmathsetmacro{\startf}{-(#1-1)/2}
        \pgfmathsetmacro{\endf}{-\startf}
        \pgfmathsetmacro{\stepf}{\startf+1}}
      \ifnum 1=#1 -- (\tikztotarget)  \else
        let \p{mid}=($(\tikztostart)!0.5!(\tikztotarget)$)
        in
        \foreach \i in {\startf,\stepf,...,\endf}
        {%
          (\tikztostart) .. controls ($ (\p{mid})!\i*6pt!90:(\tikztotarget) $) .. (\tikztotarget)
        }
      \fi
      \tikztonodes
    }
  }
}

\tikzset{circleAroundEdges/.style n args={3}{
    decorate, decoration = {markings, mark=at position #3 with {\draw[double=black,white,double distance=.6pt] (0,#2) arc [x radius = #1, y radius = #2, start angle = 90, end angle = -125];} , mark=at position #3 with{\draw[double=black,white,double distance=.6pt] (0,-#2) arc [x radius = #1, y radius = #2, start angle = 270, end angle = 360+125];}}
  }
}

\date{}

\usepackage{titling}
\usepackage[hidelinks]{hyperref}
\thanksmarkseries{arabic}

\title{Asymptotically Faster Algorithms for\\Recognizing $(k,\ell)$-Sparse Graphs}

\author{
  Bence De\'ak\thanks{Department of Operations Research, ELTE E\"otv\"os Lor\'and University, P\'azm\'any P.\ s.\ 1/c, Budapest H-1117, Hungary. E-mail: \texttt{deakbence2002@gmail.com}}
  \and P\'eter Madarasi\thanks{HUN-REN Alfr\'ed R\'enyi Institute of Mathematics, Re\'altanoda u.\ 13--15, Budapest H-1053, Hungary; and Department of Operations Research, ELTE E\"otv\"os Lor\'and University, P\'azm\'any P.\ s.\ 1/c, Budapest H-1117, Hungary. E-mail: \texttt{madarasi@renyi.hu}, corresponding author}
}

\begin{document}

\maketitle

\vspace{-10mm}
\begin{abstract}
  For fixed integers $k,\ell$ with $0\leq\ell<3k$, we study the problem of deciding whether a graph is $(k,\ell)$-sparse, a notion arising in combinatorial optimization and rigidity theory.
  Standard pebble-game algorithms take $O(n^2)$ time for $0\leq\ell<2k$ and $O(n^3)$ time for $2k\leq\ell<3k$, although faster algorithms were known for some parameter choices.

  We develop efficient recognition algorithms using bounded-indegree orientations, reductions to rooted arc-connectivity, augmenting paths, and centroid decomposition.
  For $0\leq\ell<2k$, they run in $O(n^{1+o(1)})$ time, yielding the first subquadratic bound valid throughout the entire range.
  Purely combinatorial implementations run in $O(n\sqrt n)$ time throughout this range.
  For graphs with $O(n)$ edges, we also improve the purely combinatorial running time for decomposing the edge set into a fixed number of forests, or certifying that no such decomposition exists, from $O(n\sqrt{n\log n})$ to $O(n\sqrt n)$.

  For simple graphs with $2k\leq\ell<3k$, the general recognition algorithm runs in $O(n^2)$ time when $\ell\leq2k+2$ and in $O(n^2\log n)$ time otherwise.
  When $\ell=2k$ and $k\leq3$, a specialized algorithm improves this to $O(n^{1+o(1)})$ and admits a purely combinatorial $O(n\sqrt n)$ implementation.
  Whenever the input graph is not $(k,\ell)$-sparse, our recognition algorithms return an explicit violating vertex set.

\end{abstract}


\section{Introduction}

Throughout this paper, we fix nonnegative integers $k$ and $\ell$ with $\ell<3k$.
In the \emph{classical range} $\ell < 2k$, a graph $G = (V, E)$ is called \emph{$(k, \ell)$-sparse} if, for each subset $X \subseteq V$, the number $i_G(X)$ of edges induced by $X$ satisfies $i_G(X) \leq \max\{k|X| - \ell, 0\}$.
In the extended range $2k \leq \ell < 3k$, we require this inequality only for subsets of size at least three.  
If $G$ is $(k, \ell)$-sparse and has exactly $\max\{k|V| - \ell, 0\}$ edges, then it is called \emph{$(k, \ell)$-tight}. 
A graph is \emph{$(k, \ell)$-spanning} if it contains a $(k, \ell)$-tight subgraph spanning all vertices. 
A \emph{$(k, \ell)$-block} of a $(k, \ell)$-sparse graph is a subset $X \subseteq V$ that induces a $(k, \ell)$-tight subgraph, and a \emph{$(k, \ell)$-component} is an inclusion-wise maximal $(k, \ell)$-block. 
These notions form the basis of a rich combinatorial theory closely tied to matroids and rigidity, and they naturally give rise to the optimization and recognition problems studied in this paper.

\smallskip
The concept of $(k,\ell)$-sparsity originates from Lorea~\cite{lorea} and has since surfaced in several foundational results. 
In particular, the case $\ell=k$ characterizes exactly those graphs that decompose into $k$ forests~\cite{nash1961edgeDisjointST,tutte}. More generally, for $\ell\leq k$, a graph is $(k,\ell)$-sparse precisely when its edge set can be partitioned into $\ell$ forests and $k-\ell$ pseudoforests (i.e., graphs with at most one cycle in each connected component)~\cite{gabow1988forests}.
For $(k,\ell)=(2,3)$, the $(k,\ell)$-tight graphs are the \emph{Laman graphs}, describing generically minimally rigid bar-joint frameworks in the plane~\cite{laman}, while $(2,3)$-spanning graphs correspond to rigid frameworks~\cite{jordan2016combinatorial}.
In the classical range, the family of $(k,\ell)$-sparse edge sets forms the independent sets of a matroid, allowing efficient optimization over $(k, \ell)$-sparse subgraphs~\cite{lorea,lee2008pebble}.
In the extended range $2k \leq \ell < 3k$, three-dimensional rigidity provides further motivation: the boundary case $\ell=2k$ plays a key role in the analysis of block--hole graphs with a single block and yields necessary conditions for 3D rigidity when $k=3$~\cite{jordan2023RigidBlockAndHoleGraphs}. 
Altogether, $(k,\ell)$\nobreakdash-sparsity offers a unifying framework linking matroid theory, edge-disjoint spanning structures, and planar as well as spatial rigidity.

\smallskip
In this paper, we consider the \emph{recognition problem}, which asks whether a given graph is $(k,\ell)$-sparse. 
In the classical range, this is precisely the independence oracle for the matroid of $(k,\ell)$-sparse edge sets and thus serves as a crucial building block for numerous combinatorial algorithms, especially in rigidity theory. Throughout the paper, we treat $k$ and $\ell$ as fixed constants and therefore omit $\mathrm{poly}(k, \ell)$ factors from asymptotic bounds. We denote by $n$ and $m$ the numbers of vertices and edges of any input graph, respectively. In the recognition problem, we restrict attention throughout to instances with $m \leq kn$, as any graph with $m>kn$ trivially violates the sparsity condition.

\subsection{Historical algorithms}

For recognizing $(k, \ell)$-sparse graphs, the classical approach is the \emph{pebble game} framework introduced in~\cite{bergPhD,berg2003algorithms,hendrickson,lee2008pebble,pebbleDS}.
The pebble game is a greedy algorithm exploiting the matroidal structure of $(k,\ell)$-sparse graphs: it processes the edges one by one and inserts an edge whenever the resulting subgraph remains $(k,\ell)$-sparse. 
Testing whether an edge can be added is based on the Orientation Lemma~\cite{SLHOrientationLemma}, by maintaining a $k$-indegree-bounded orientation and searching for augmenting paths. In its straightforward implementation, the running time of the pebble game algorithm is $O(n^2)$.
Although several faster algorithms were developed for special choices of $k$ and $\ell$, which we review below, it remained open whether recognition throughout the entire classical range could be achieved in subquadratic time.

\paragraph{The range $\ell \leq k$.} For $\ell \leq k$, a graph is $(k,\ell)$-sparse if and only if its edge set can be decomposed into $\ell$ forests and $k-\ell$ pseudoforests~\cite{gabow1988forests}. This yields a recognition algorithm due to Gabow and Westermann~\cite{gabow1988forests} that maintains such a decomposition for the accepted edges. The algorithm first applies a preprocessing step that deletes vertices of degree at most $k$, storing the incident edges so that they can be restored in a postprocessing step. It then performs an initial augmentation phase using the subroutine \emph{batch}. After these batch steps, the remaining edges are processed one at a time: for each such edge, the algorithm first tries to insert it directly into one of the maintained parts. If this is not possible, it searches for an augmenting sequence of exchanges among the parts. The crucial point is that this search is implemented using specialized scan procedures rather than a generic matroid-union routine: the algorithm uses \emph{breadth-first scanning} for the pseudoforest parts and \emph{cyclic scanning} for the forest parts. If these scans find an augmentation, then the corresponding exchanges are performed and the edge is accepted; otherwise, the edge is rejected. This algorithm decides independence in the union of $k-\ell$ bicircular and $\ell$ graphic matroids, thereby yielding a test for $(k,\ell)$-sparsity. The same framework can also be extended to recognize $(k,k+1)$-sparse graphs --- in particular, $(2,3)$-sparse graphs --- by first finding a maximum-size $(k,k)$-sparse subgraph decomposed into $k$ forests. If this subgraph does not contain all input edges, then the graph is not $(k,k+1)$-sparse. Otherwise, one checks whether the subgraph contains a nonempty \emph{top clump}, i.e., an inclusion-wise maximal \emph{clump}. A top clump can be found efficiently from the computed decomposition. In this case, nonempty clumps correspond to edge sets of $(k,k)$-tight subgraphs, so a $(k,k)$-sparse graph is $(k,k+1)$-sparse if and only if no nonempty top clump exists; otherwise, the vertex set of the corresponding $(k,k)$-tight subgraph violates $(k,k+1)$-sparsity.

\paragraph{Laman graphs.} For Laman graphs, a more efficient alternative to naive recognition methods is based on the fact that a graph is Laman if and only if, after duplicating any edge, the resulting multigraph can be partitioned into two edge-disjoint spanning trees. This decomposition provides the red and black trees on which the notion of a \emph{red-black hierarchy} is built, giving an equivalent characterization: a graph is Laman if and only if it admits such a hierarchy. In the algorithm introduced in~\cite{bereg2005certifying}, one first duplicates an edge, computes the two edge-disjoint spanning trees, performs a decomposition according to a specific set of rules, and then constructs the graph corresponding to that decomposition in $O(n^2)$ time. A key observation is that the graph characterizing the decomposition is a red-black hierarchy if and only if all edges are removed during the decomposition process~\cite{daescu2009towards}; therefore, instead of explicitly constructing the decomposition graph, it suffices to test whether all edges were deleted. Given the two edge-disjoint spanning trees, the improved algorithm carries out this decomposition in $O(n\log n)$ time using tree linearization via DFS numbering together with segment trees to report and delete crossing~edges.

\paragraph{Outside the classical range.} In the range $2k \leq \ell < 3k$, the standard approach is a straightforward extension of the pebble game algorithm. 
When deciding whether an edge $uv$ can be inserted, one has to check, for every vertex $w \notin \{u,v\}$, whether there exists a $(k,\ell)$-tight subgraph containing all three vertices $u,v,w$. 
Performing this test separately for each choice of $w$ leads to an $O(n^3)$-time algorithm~\cite{madarasi2023klSparse}. 
In the special case $\ell = 2k$, this verification can be carried out for all vertices $w$ simultaneously using a constant number of graph searches per input edge, reducing the overall running time to $O(n^2)$~\cite{madarasi2023klSparse}.

\subsection{Our contribution}

We design new algorithms for all three ranges $\ell \leq k$, $k<\ell<2k$, and $2k \leq \ell < 3k$.
In the range $k < \ell < 2k$, we assume that the input graph is loopless, as any loop would trivially violate the sparsity condition.
For $2k \leq \ell < 3k$, we follow the standard convention and assume that the input graph is simple.
A central ingredient of our approach is a novel reduction of a key subproblem of recognition to rooted arc-connectivity in directed graphs, which can be solved very efficiently for the relevant parameters.
We combine this with flow-based and augmenting-path techniques and, in the most technically challenging range $k<\ell<2k$, with a divide-and-conquer method based on centroid decomposition.
Across these parameter ranges, we obtain new asymptotic bounds; in particular, our algorithms are the first to achieve subquadratic recognition throughout the entire classical range.
Under purely combinatorial implementations, our running time is $O(n\sqrt n)$ throughout the classical range $0 \leq \ell < 2k$.
At the boundary $\ell=2k$, our specialized algorithm runs in near-linear time for $k\leq3$; for $k=2,3$, it also admits a purely combinatorial $O(n\sqrt n)$ implementation.

As a separate algorithmic consequence, for every fixed positive integer $\kappa$, Proposition~\ref{proposition:forest_decomp_combinatorial} gives a purely combinatorial $O(n\sqrt n)$-time algorithm that, on a graph with $O(n)$ edges, either decomposes the edge set into $\kappa$ forests or returns a vertex set certifying that no such decomposition exists, improving the $O(n\sqrt{n\log n})$ bound of Gabow and Westermann~\cite{gabow1988forests}.

A notable special case of our recognition result is $\ell = k+1$ --- and in particular the fundamental case $k=2$, $\ell=3$ from combinatorial rigidity theory --- for which Gabow and Westermann gave an $O(n\sqrt{n\log n})$ algorithm~\cite{gabow1988forests}.
As noted in~\cite{daescu2009towards}, it is not clear whether this bound can be improved solely by using a faster forest decomposition procedure.
The main contribution of~\cite{daescu2009towards} is precisely to decouple these aspects in the Laman graph case, where their algorithm recognizes Laman graphs by using a forest decomposition algorithm as a black box in an initial step, thereby separating the decomposition from the subsequent computations.
However, this approach is specific to Laman graphs, does not extend to general $(2,3)$-sparse graphs, and cannot produce a certificate of failure in the form of a vertex set violating the sparsity condition.
In contrast, our algorithm works throughout the entire range $k < \ell < 2k$, retains the decoupling of the forest decomposition from the later stages of the computation, and can also return a violating set when the input graph is not $(k, \ell)$-sparse.

\paragraph{Summary of running times.}
Table~\ref{tab:complexities} summarizes the best asymptotic running-time bounds achieved for the recognition problem in each parameter range. More detailed bounds --- given in terms of the running-time notation introduced later in Section~\ref{sec:alg_primitives} --- are presented in Table~\ref{tab:recognition_complexities}.

\begin{table}[H]
\centering
\caption{Asymptotic running-time bounds for the main problems.}
\label{tab:complexities}
\setlength{\tabcolsep}{8pt}
\renewcommand{\arraystretch}{1.28}
\begin{threeparttable}
\begin{tabularx}{\textwidth}{@{} >{\raggedright\arraybackslash}m{.25\textwidth}
                                    >{\raggedright\arraybackslash}m{.435\textwidth}
                                    >{\raggedright\arraybackslash}X @{}}
\toprule
Range & Old bounds & New bounds \\
\midrule\addlinespace[6pt]

$0 \leq \ell \leq k$
& $\begin{cases}
O(n\sqrt{n}) & \text{if } \ell = 0~\cite{gabow1988forests} \\
O(n\sqrt{n \log n}) & \text{if } 0 < \ell < k~\cite{gabow1988forests} \\
O(n^{1 + o(1)})\tnote{*} & \text{if } \ell = k~\cite{arkhipov2026faster}
\end{cases}$
& \vspace{13pt}$O(n^{1+o(1)})$\tnote{†}
\\

\addlinespace[6pt]\midrule\addlinespace[6pt]

$k < \ell < 2k$
& $\begin{cases}
O(n\sqrt{n \log n})\tnote{‡} & \text{if } \ell = k + 1~\cite{gabow1988forests} \\
O(n^{2}) & \text{if } \ell > k + 1~\cite{lee2008pebble}
\end{cases}$
& \vspace{6pt}$O(n^{1+o(1)})$\tnote{†}
\\

\addlinespace[6pt]\midrule\addlinespace[6pt]

$\ell = 2k$
& $\begin{cases}
O(n) & \text{if } k = 1 \\
O(n^2) & \text{if } k > 1~\cite{madarasi2023klSparse}
\end{cases}$
& $\begin{cases}
O(n^{1 + o(1)})\tnote{§} & \text{if } k \leq 3 \\
O(n^2) & \text{if } k > 3
\end{cases}$
\\

\addlinespace[6pt]\midrule\addlinespace[6pt]

$2k < \ell < 3k$
& $O(n^3)$~\cite{madarasi2023klSparse}
& $\begin{cases}
O(n^2) & \text{if } \ell \leq 2k + 2 \\
O(n^2 \log n) & \text{if } \ell > 2k + 2
\end{cases}$
\\

\addlinespace[2pt]
\bottomrule
\end{tabularx}

\begin{tablenotes}
\setlength{\itemsep}{4pt}
\item[*] The best-known purely combinatorial algorithm has worst-case running time $O(n \sqrt{n \log n})$.
\item[†] Under purely combinatorial implementations of our primitives, the new bound is $O(n\sqrt{n})$ throughout the classical range $0 \leq \ell < 2k$.
\item[‡] In the special case of Laman graphs, i.e., for $(k,\ell)=(2,3)$ with $m=2n-3$, an $O(T_{\mathrm{st}}(n)+n\log n)$-time recognition algorithm is known, where $T_{\mathrm{st}}(n)$ is the time required to extract two edge-disjoint spanning trees~\cite{daescu2009towards}. Combined with the near-linear forest-packing algorithm given in~\cite{arkhipov2026faster}, this yields an $O(n^{1+o(1)})$ bound. For the planar Laman case, a more specialized algorithm was proposed in~\cite{rollin2019recognizing} with a time complexity of $O(n \log^3 n)$. However, neither algorithm provides a violating vertex set certifying that the input graph is not $(2, 3)$-sparse.
\item[§] Under purely combinatorial implementations of our primitives, the new bound for $\ell=2k$ is $O(n\sqrt n)$ when $k\leq3$.
\end{tablenotes}

\end{threeparttable}
\end{table}

\section{Structural and algorithmic preliminaries}

\subsection{Running-time primitives and complexity overview}
\label{sec:alg_primitives}

The recognition algorithms developed in the following subsections rely on several standard graph-theoretic subroutines, which we treat as black boxes. We first specify the input and output behavior of each subroutine and then introduce notation for its running time, summarized in Table~\ref{tab:parameters}.

For a digraph $D=(V,A)$ and a vertex set $X\subseteq V$, let $i_D(X)$ denote the number of arcs with both endpoints in $X$, and let $\varrho_D(X)$ denote the number of arcs entering $X$ from $V\setminus X$.
We write $\varrho_D(v)=\varrho_D(\{v\})$ for the indegree of a vertex $v$.
Given a root $s\in V$ and a nonnegative integer $\eta$, we say that $D$ is \emph{rooted $\eta$-arc-connected} with root $s$ if $\varrho_D(X)\geq\eta$ for every nonempty set $X\subseteq V\setminus\{s\}$.
Given a graph $G=(V,E)$ and a function $b:V\to\mathbb{Z}_{\geq0}$, an orientation $D$ of $G$ is $b$-indegree-bounded if $\varrho_D(v)\leq b(v)$ for every $v\in V$.
When $b(v)=\kappa$ for every $v\in V$, we say that $D$ is $\kappa$-indegree-bounded.

\begin{spec}{MaxFlow}{$D=(V,A),s,t,g$}{spec:maxflow}
{A digraph $D=(V,A)$ with $n$ vertices and $O(n)$ arcs, a source $s$, a sink $t$, and a capacity function $g:A\to\mathbb{Z}_{\geq 0}$ such that the sum of capacities is $O(n)$.}
{A pair $(x,S)$, where $x:A\to\mathbb{Z}_{\geq 0}$ is a maximum $g$-feasible flow and $s \in S \subseteq V \setminus \{t\}$ is a minimum $s$-$t$ cut.}
{The best-known algorithm for the maximum flow problem~\cite{brand2023deterministic} is based on an interior-point method and has near-linear time complexity, i.e., it is asymptotically faster than $O(n^{1+\varepsilon})$ for any $\varepsilon > 0$. A purely combinatorial alternative is Dinic's algorithm~\cite{dinic1970algorithm}, whose running time is $O(n \sqrt{n})$ in the special case when the number of arcs and the sum of capacities are both $O(n)$~\cite{even1975network}.}
\end{spec}

\begin{spec}{BoundedOri$_{\kappa}$}{$G=(V,E)$}{spec:boundedori}
{A graph $G=(V,E)$ with $n$ vertices and $O(n)$ edges.}
{A pair $(X, D)$, where $X=\varnothing$ and $D$ is a $\kappa$-indegree-bounded orientation if such an orientation exists; otherwise, $X$ is a vertex set that violates the $(\kappa, 0)$-sparsity of $G$, and $D = \varnothing$.}
{We can find a $\kappa$-indegree-bounded orientation in near-linear time by reducing it to a maximum-flow problem as shown below.}
\end{spec}

For the sake of completeness, we present a simple reduction from the $\kappa$-indegree-bounded orientation problem to a maximum-flow problem in the following lemma.

\begin{lemma}
\label{lemma:k_orientation}
Given a graph $G=(V,E)$ with $n$ vertices and $O(n)$ edges, the problem of finding a $\kappa$-indegree-bounded orientation --- or certifying its nonexistence by producing a violating vertex set --- can be reduced in linear time to a maximum-flow instance with $O(n)$ vertices, $O(n)$ arcs, and total capacity $O(n)$.
\end{lemma}

\begin{proof}
Let $m=|E|$.
If $m>\kappa n$, then $i_G(V)=m>\kappa|V|$, so $V$ violates $(\kappa,0)$-sparsity and can be returned immediately.
Hence, we may assume that $m\leq \kappa n$.

We construct a flow network with source $s$, sink $t$, the vertices of $V$, and one additional vertex $z_e$ for each edge $e\in E$.
For every edge $e=uv$, add the arcs $sz_e$, $z_eu$, and $z_ev$, each with capacity $1$.
For every vertex $v\in V$, add the arc $vt$ with capacity $\kappa$.
Let $(x,S)$ be the output of \textsc{MaxFlow} on this network.

Suppose first that $x$ has value $m$.
Then every arc $sz_e$ is saturated.
Since $x$ is integral, flow conservation at $z_e$ implies that exactly one of the arcs $z_eu$ and $z_ev$ carries one unit of flow.
Orient $e$ toward the corresponding endpoint.
For every vertex $v$, the flow on $vt$ equals the number of edges oriented toward $v$ and is therefore at most $\kappa$.
Thus, the resulting orientation is $\kappa$-indegree-bounded.
Conversely, given a $\kappa$-indegree-bounded orientation of $G$, send one unit of flow along $sz_e$, $z_ev$, and $vt$ for every edge $e$ oriented toward $v$.
The capacity of $vt$ is respected because at most $\kappa$ edges are oriented toward $v$, so this defines a flow of value $m$.

Suppose now that $x$ has value less than $m$, and set $X=S\cap V$.
Consider an edge $e=uv$ that is not induced by $X$.
If $z_e\notin S$, then the arc $sz_e$ crosses the cut.
If $z_e\in S$, then at least one endpoint of $e$ lies outside $X$, so at least one of the arcs $z_eu$ and $z_ev$ crosses the cut.
Therefore, the cut contains at least one capacity-$1$ arc for each of the $m-i_G(X)$ edges not induced by $X$.
Moreover, the arc $vt$ crosses the cut for every $v\in X$, contributing $\kappa|X|$ to its capacity.
Hence, the capacity of the cut is at least $m-i_G(X)+\kappa|X|$.
By the max-flow min-cut theorem, this capacity equals the value of $x$ and is less than $m$.
It follows that $i_G(X)>\kappa|X|$, so $X$ violates the $(\kappa,0)$-sparsity of~$G$.

The network has $n+m+2$ vertices, $3m+n$ arcs, and total capacity $3m+\kappa n$.
Since $\kappa$ is fixed and $m\leq \kappa n$, all three quantities are $O(n)$.
The network can be constructed in linear time, and either the orientation or the violating set can be recovered from the output in linear time.
\end{proof}

\begin{spec}{ForestDecomp$_{\kappa}$}{$G=(V,E)$}{spec:forestdecomp}
{A graph $G=(V, E)$ with $n$ vertices and $O(n)$ edges.}
{A tuple $(X, F_1, \dots, F_{\kappa})$, where $X=\varnothing$ and $(F_1, \dots, F_{\kappa})$ is a partition of $E$ into spanning forests if such a partition exists; otherwise, $X$ is a vertex set that violates the $(\kappa, \kappa)$-sparsity of $G$, and $F_1, \dots, F_{\kappa} =~\varnothing$.}
{For computing a forest decomposition, the fastest algorithm currently known~\cite{arkhipov2026faster} uses the near-linear maximum-flow subroutine mentioned above. The Gabow--Westermann algorithm~\cite{gabow1988forests} gives a purely combinatorial $O(n\sqrt{n\log n})$ bound. Proposition~\ref{proposition:forest_decomp_combinatorial} improves this bound to $O(n\sqrt n)$ by combining \textsc{BoundedOri}$_{\kappa}$ with the bounded-indegree directed $\kappa$-forest algorithm of Arkhipov and Kolmogorov~\cite{arkhipov2026faster}.}
\end{spec}

\begin{note}
\label{note:forest_decomp_violating_set}
Let $F_1,\dots,F_{\kappa}$ be $\kappa$ disjoint forests of maximum total size that do not partition $E$, and let $H=(V,F_1\cup\dots\cup F_{\kappa})$.
For any $e\in E\setminus E(H)$, the theorem of Nash-Williams~\cite{nash1961edgeDisjointST} implies that $H+e$ is not $(\kappa,\kappa)$-sparse.
Orienting each component of each $F_i$ away from an arbitrary root gives a $\kappa$-indegree-bounded orientation of $H$, so one iteration of the pebble game algorithm~\cite{lee2008pebble} applied to $e$ yields a violating vertex set in $O(n)$ time.
\end{note}

\begin{spec}{BoundedForestDecomp$_{\kappa}$}{$D=(V,A)$}{spec:boundedforestdecomp}
{A $\kappa$-indegree-bounded digraph $D=(V,A)$ with $n$ vertices and $O(n)$ arcs; let $G=(V,E)$ denote its underlying graph.}
{A tuple $(X, F_1, \dots, F_{\kappa})$, where $X=\varnothing$ and $(F_1, \dots, F_{\kappa})$ is a partition of $E$ into spanning forests if such a partition exists; otherwise, $X$ is a vertex set that violates the $(\kappa, \kappa)$-sparsity of $G$, and $F_1, \dots, F_{\kappa} =~\varnothing$.}
{Apply the purely combinatorial bounded-indegree directed $\kappa$-forest algorithm of Arkhipov and Kolmogorov~\cite{arkhipov2026faster} to $D$.
This returns $\kappa$ edge-disjoint forests $F_1,\dots,F_{\kappa}$ of maximum total size subject to the condition that the indegree of their union is at most $\kappa$ at every vertex.
For fixed $\kappa$, this takes $O(n\log n)$ time on the present $O(n)$-arc instance.
Since every subgraph of $D$ also has indegree at most $\kappa$ at every vertex, this condition is vacuous on this instance.
Thus, $F_1,\dots,F_{\kappa}$ have maximum total size among all $\kappa$ edge-disjoint forests of $G$.
If they partition $E$, we are done.
Otherwise, Note~\ref{note:forest_decomp_violating_set} gives a vertex set that violates $(\kappa,\kappa)$-sparsity in $O(n)$ time.}
\end{spec}

\begin{proposition}
\label{proposition:forest_decomp_combinatorial}
\textsc{ForestDecomp}$_{\kappa}$ has a purely combinatorial $O(n\sqrt n)$ implementation.
\end{proposition}

\begin{proof}
First run \textsc{BoundedOri}$_{\kappa}$ on $G$.
If it returns a violating vertex set $X$, then $i_G(X)>\kappa|X|$, so $X$ also violates $(\kappa,\kappa)$-sparsity and can be returned.
Otherwise, let $D$ be the resulting $\kappa$-indegree-bounded orientation and run \textsc{BoundedForestDecomp}$_{\kappa}$ on $D$.
By its specification, this either returns a decomposition of $E$ into $\kappa$ forests or a vertex set violating $(\kappa,\kappa)$-sparsity.
By Lemma~\ref{lemma:k_orientation} and the purely combinatorial maximum-flow bound stated above, \textsc{BoundedOri}$_{\kappa}$ takes $O(n\sqrt n)$ time.
The \textsc{BoundedForestDecomp}$_{\kappa}$ call takes $O(n\log n)$ time by its specification, so the total running time is $O(n\sqrt n)$.
\end{proof}

\begin{spec}{RootedTermConn$_{\eta}$}{$D=(V,A)$, $s$, $Z$}{spec:rootedterminalarcconn}
{A digraph $D=(V,A)$ with $n$ vertices and $O(n)$ arcs, a root $s\in V$, and a set $Z\subseteq V\setminus\{s\}$ of terminals.}
{A nonempty set $X\subseteq V\setminus\{s\}$ such that $X\cap Z\neq\emptyset$ and $\varrho_D(X)<\eta$, if such a set exists; otherwise $\varnothing$.}
{We can implement this subroutine as follows.
\begin{itemize}[itemindent=12pt,labelsep=4pt]
\item When $\eta = 0$, there is nothing to check.
\item When $\eta = 1$, perform a single graph search in linear time to obtain the set $S$ of vertices that are reachable from the root. If $Z \not \subseteq S$, then $V \setminus S$ is a violating vertex set; otherwise, no violating set exists.
\item When $\eta = 2$, first run the $\eta=1$ test, and let $S$ be the set of vertices reachable from $s$. If all vertices in $Z$ are reachable, subdivide every arc of $D[S]$ and compute a dominator tree~\cite{alstrup1999dominators}; an arc $a=uv$ of $D[S]$ is a bridge exactly when its subdivision vertex dominates $v$.
Deleting a bridge whose corresponding subtree contains at least one element from $Z$ yields a violating set $X$ consisting of the original vertices in that subtree together with $V\setminus S$, while the absence of such a bridge certifies that no violating vertex set exists.
This takes $O(n+m)=O(n)$ time.
\item When $\eta = 3$, let $D'$ be obtained from $D$ by adding three parallel arcs from each $v \in Z \setminus \{s\}$ to $s$, and compute the $3$-edge-connected components of $D'$ in linear time~\cite{georgiadis2024three}. If there is a vertex $v \in Z$ that is not in the same $3$-edge-connected component as $s$, then the added arcs guarantee three arc-disjoint paths from $v$ to $s$ and do not affect any $s$-$v$ cut. Hence, run a unit-capacity $s$-$v$ flow computation in $D$, stopping once the flow value reaches $3$. If the maximum flow has value below $3$, the corresponding residual cut is a violating vertex set. Otherwise, no violating set exists. The total running time is $O(n+m)=O(n)$.
\item When $\eta \geq 4$, for each $v \in Z \setminus \{s\}$, run a unit-capacity $s$-$v$ flow computation in $D$, stopping once the flow value reaches $\eta$. If any flow has value below $\eta$, its residual cut is a violating vertex set; if every flow reaches value $\eta$, then no violating set exists. Since $\eta$ is fixed, the $O(n)$ flow computations take a total of $O(n(n+m))=O(n^2)$ time.
\end{itemize}
}
\end{spec}

\begin{spec}{RootedArcConn$_{\eta}$}{$D=(V, A)$, $s$}{spec:rootedarcconn}
{A digraph $D=(V, A)$ with $n$ vertices and $O(n)$ arcs, and a root $s \in V$.}
{A nonempty set $X \subseteq V \setminus \{s\}$ such that $\varrho_D(X) < \eta$, if such a set exists; otherwise $\varnothing$.}
{We can implement this subroutine as follows.
\begin{itemize}[itemindent=12pt,labelsep=4pt]
\item When $\eta \leq 3$, we call \textsc{RootedTermConn$_{\eta}$}($D$, $s$, $V \setminus \{s\}$), which runs in linear time.
\item When $\eta \geq 4$, add $\eta$ parallel arcs from each $v\in V\setminus\{s\}$ to $s$.
In the augmented digraph, the cuts of value below $\eta$ correspond exactly to the nonempty sets $X\subseteq V\setminus\{s\}$ with $\varrho_D(X)<\eta$.
Using Gabow's algorithm, we can either find such a cut or certify that none exists in $O(n\log n)$ time~\cite{gabow1991matroid}.
\end{itemize}}
\end{spec}

Table~\ref{tab:parameters} summarizes the running times of our subroutines.
For the sake of precision, we fix positive constants $C_1,\ldots,C_6$ large enough such that every instance of ``size'' $n$ considered henceforth --- including the maximum-flow network of Lemma~\ref{lemma:k_orientation} --- meets the corresponding size bounds in Table~\ref{tab:parameters}. The stated running-time bounds reflect the best currently known algorithms.

\begin{table}[H]
\centering
\caption{Notation and running-time bounds for standard subroutines.}
\setlength{\tabcolsep}{7pt}
\renewcommand{\arraystretch}{1.28}
\begin{threeparttable}

\begin{tabularx}{\textwidth}{@{}
  >{\raggedright\arraybackslash}m{.11\textwidth}
  >{\raggedright\arraybackslash}m{.36\textwidth}
  >{\raggedright\arraybackslash}m{.18\textwidth}
  >{\RaggedRight\arraybackslash}X
@{}}
\toprule
Notation & Subroutine & Constraints & Best bound \\
\midrule\addlinespace[6pt]

$T_{MF}(n)$
& \textsc{MaxFlow}($D=(V,A), s, t, g$)
& $|V|, |A|, g(A) \leq C_1 n$
& $O(n^{1+o(1)})$\tnote{*}~\cite{brand2023deterministic}
\\

\addlinespace[6pt]\midrule\addlinespace[6pt]

$T_{BO}(n, \kappa)$
& \textsc{BoundedOri$_{\kappa}$}($G=(V,E)$)
& $|V|, |E| \leq C_2 n$
& $O(T_{MF}(n)) \subseteq O(n^{1+o(1)})$\tnote{*}
\\

\addlinespace[6pt]\midrule\addlinespace[6pt]
$T_{BFD}(n, \kappa)$
& \textsc{BoundedForestDecomp$_{\kappa}$}($D=(V,A)$)
& $|V|, |A| \leq C_3 n$
& $O(n \log n)$~\cite{arkhipov2026faster}
\\

\addlinespace[6pt]\midrule\addlinespace[6pt]
$T_{FD}(n, \kappa)$
& \textsc{ForestDecomp$_{\kappa}$}($G=(V,E)$)
& $|V|, |E| \leq C_4 n$
& $O(T_{BO}(n, \kappa) + T_{BFD}(n,\kappa)) \subseteq O(n^{1+o(1)})$\tnote{*}~\cite{arkhipov2026faster}
\\

\addlinespace[6pt]\midrule\addlinespace[6pt]

\parbox[c]{\linewidth}{$T_{RTC}(n,\eta)$}
& \parbox[c]{\linewidth}{\textsc{RootedTermConn}$_{\eta}$($D=(V,A), s,Z$)}
& \parbox[c]{\linewidth}{$|V|, |A| \leq C_5 n$}
& \parbox[c]{\linewidth}{$\begin{cases}
  O(n) & \text{if } \eta \leq 3~\cite{alstrup1999dominators,georgiadis2024three} \\
  O(n^2) & \text{if } \eta > 3
\end{cases}$}

\\

\addlinespace[6pt]\midrule\addlinespace[6pt]

\parbox[c]{\linewidth}{$T_{RC}(n,\eta)$}
& \parbox[c]{\linewidth}{\textsc{RootedArcConn}$_{\eta}$($D=(V,A), s$)}
& \parbox[c]{\linewidth}{$|V|, |A| \leq C_6 n$}
& \parbox[c]{\linewidth}{$\begin{cases}
  O(n) & \text{if } \eta \leq 3~\cite{alstrup1999dominators,georgiadis2024three} \\
  O(n \log n) & \text{if } \eta > 3~\cite{gabow1991matroid}
\end{cases}$
\vspace{2pt}}
\\

\bottomrule
\end{tabularx}

\begin{tablenotes}
\item[*] If we restrict ourselves to purely combinatorial algorithms, then the best bounds for the first and second rows change to $O(n\sqrt{n})$~\cite{dinic1970algorithm,even1975network}. By Proposition~\ref{proposition:forest_decomp_combinatorial}, the fourth row also admits an $O(n\sqrt n)$ purely combinatorial implementation.
\end{tablenotes}

\end{threeparttable}

\label{tab:parameters}
\end{table}

\begin{note}
In the constructions used in this paper, the choices $C_1 = 4k, C_2 = k, C_3=k$, $C_4=k$, $C_5=2k$ and $C_6=2k$ are adequate as rough estimates. Note, however, that the asymptotic bounds above do not depend on the particular values of the constants $C_i$. Accordingly, in our later constructions, we only show that any quantity we bound by $C_i n$ --- such as the number of vertices, the number of edges, or the total capacity --- is $O(n)$.
\end{note}

\begin{note}
In the worst case, all of our subroutines above must read their entire input. The only exception is \textsc{RootedArcConn}$_{\eta}$ for $\eta = 0$, where the output is always $\varnothing$, regardless of the input. Nevertheless, for simplicity, we impose the same assumption in this case as well and therefore assume $T_{RC}(n,0) = \Theta(n)$.
\end{note}

\paragraph{A convenient regularization.}
For technical convenience, we also define
\[
T'_{RC}(n,\eta)=n\cdot \max \left\{\frac{T_{RC}(q,\eta)}{q}: q \in \{1, \dots, n\}\right\} \;\text{ and }\;
T'_{RTC}(n,\eta)=n\cdot \max \left\{\frac{T_{RTC}(q,\eta)}{q}: q \in \{1, \dots, n\}\right\}.
\]
By construction, $T_{RC}(n,\eta) \leq T'_{RC}(n,\eta)$ holds for all $n$; hence, $T_{RC}(n,\eta)=O(T'_{RC}(n,\eta))$. Moreover, $T'_{RC}(n,\eta)/n$ is nondecreasing in $n$, and $T'_{RC}$ satisfies the same asymptotic bounds as $T_{RC}$ listed in Table~\ref{tab:parameters}. The same holds for $T_{RTC}$ and $T'_{RTC}$.
\paragraph{Summary of running times.}
Table~\ref{tab:recognition_complexities} summarizes the asymptotic running times of our new algorithms, developed in Section~\ref{sec:recognition_algorithms}.
The bounds are expressed using the running-time primitives introduced above.
For completeness, the table also shows the asymptotic bounds obtained by instantiating these primitives with the fastest known implementations.

\begin{table}[H]
\centering
\caption{Asymptotic running-time bounds for the recognition problem.}
\label{tab:recognition_complexities}
\setlength{\tabcolsep}{8pt}
\renewcommand{\arraystretch}{1.28}
\begin{threeparttable}

\begin{tabularx}{\textwidth}{@{}
  >{\raggedright\arraybackslash}m{.33\textwidth}
  >{\raggedright\arraybackslash}X
@{}}
\toprule
Range & New bounds \\
\addlinespace[6pt]\midrule\addlinespace[6pt]

\parbox[c]{\linewidth}{$0 \leq \ell \leq k$}
& \parbox[c]{\linewidth}{$O(T_{BO}(n,k) + T_{RC}(n,\ell)) \nobreak\subseteq {} O(n^{1+o(1)})$\tnote{*}}
\\

\addlinespace[6pt]\midrule\addlinespace[6pt]

\parbox[c]{\linewidth}{$k < \ell < 2k$}
& \parbox[c]{\linewidth}{$O(T_{FD}(n, k) + T'_{RC}(n,\ell-k)\log n) \nobreak\subseteq {} O(n^{1+o(1)})$\tnote{*}}
\\

\addlinespace[6pt]\midrule\addlinespace[6pt]

\parbox[c]{\linewidth}{$\ell = 2k$}
& \parbox[c]{\linewidth}{$
\begin{cases}
O(T_{FD}(n,k)+T'_{RTC}(n,k)\log n) \nobreak\subseteq {} O(n^{1+o(1)})\tnote{†} & \text{if } k\leq2,\\
O(T_{FD}(n,k)+T'_{RTC}(n,k)\log n + n \log^2 n) \nobreak\subseteq {} O(n^{1+o(1)})\tnote{†} & \text{if } k=3,\\
O(n^2) & \text{if } k>3.
\end{cases}$
}
\\

\addlinespace[6pt]\midrule\addlinespace[6pt]

\parbox[c]{\linewidth}{$2k < \ell < 3k$}
& \parbox[c]{\linewidth}{$O(n\cdot T_{RC}(n,\ell+1-2k)) \nobreak\subseteq {}
\begin{cases}
O(n^2) & \text{if } \ell \leq 2k+2\\
O(n^2\log n) & \text{if } \ell > 2k+2
\end{cases}$ \vspace{2pt}}
\\

\bottomrule
\end{tabularx}

\begin{tablenotes}
\setlength{\itemsep}{4pt}
\item[*] Under purely combinatorial implementations of our primitives, the new bound is $O(n\sqrt{n})$ throughout the classical range $0 \leq \ell < 2k$.
\item[†] Under purely combinatorial implementations of the displayed subroutines, the specialized algorithm for $\ell=2k$ runs in $O(n\sqrt n)$ time when $k\leq3$.
\end{tablenotes}

\end{threeparttable}
\end{table}

\subsection{Structural preliminaries}

In this subsection, we establish the structural properties and algorithmic ingredients common to all parameter ranges. A key lemma shows that an essential subtask of the recognition problem reduces to testing rooted $\eta$-arc-connectivity. Since for small fixed $\eta$, this can be performed very efficiently (see Table~\ref{tab:parameters}), the subtask itself also admits a fast solution. We begin with a simple definition.

\begin{definition}
Given a digraph $D=(V, A)$ and an independent vertex set $U_0 \subseteq V$, we say that $D$ is $U_0$-source if $\varrho_D(v) = 0$ for all $v \in U_0$.
\end{definition}

\begin{lemma}
\label{lemma:u0_source}
Let $t$ be a nonnegative integer, let $G = (V, E)$ be a graph, and let $U_0 \subseteq V$ be an independent set of size~$t$. Then $G$ has a $k$-indegree-bounded $U_0$-source orientation if and only if every strict superset $X \subseteq V$ of $U_0$ satisfies the $(k, tk)$-sparsity bound in~$G$.
\end{lemma}

\begin{proof}
The condition is clearly necessary, as for a $k$-indegree-bounded $U_0$-source orientation $D_0$ and any vertex set $X$ strictly containing $U_0$, we have
\[
  i_G(X) \leq \sum_{v \in X} \varrho_{D_0}(v) = \sum_{v \in X \setminus U_0} \varrho_{D_0}(v) \leq k|X \setminus U_0| = k|X| - tk.
\]
To prove sufficiency, we will apply the Orientation Lemma~\cite{SLHOrientationLemma}, which states that for a given upper-capacity function $g: V \to \mathbb{Z}_{\geq 0}$, there exists a $g$-indegree-bounded orientation of $G$ if and only if
\[
  i_G(X) \leq \sum_{v \in X} g(v) \eqqcolon g(X)
\]
holds for every vertex set $X \subseteq V$.
Define $g$ by setting $g(v)=0$ for $v\in U_0$ and $g(v)=k$ for $v\notin U_0$. Fix any $X$ and let $r = |X \setminus U_0|$. If $r=0$, then $i_G(X) = 0 = g(X)$ since $U_0$ is independent. Otherwise, by the sparsity condition,
\[
i_G(X) \leq i_G(X \cup U_0) \leq (t + r)k - tk
= rk = k \cdot |X \setminus U_0| + 0 \cdot |X \cap U_0|
= g(X).
\]
Hence the desired orientation $D_0$ exists.
\end{proof}

\begin{lemma}
\label{lemma:main}
Let $t$ be a nonnegative integer such that $tk \leq \ell \leq (t+1)k$. Let $G=(V,E)$ be a graph, and let $U_0 \subseteq V$ be an independent set of size $t$. Suppose that $G$ has a $k$-indegree-bounded $U_0$-source orientation $D_0$. Let $D_0'$ be the digraph obtained from $D_0$ by adding a new root vertex $s$, deleting the vertices in $U_0$, and, for each $v \in V \setminus U_0$, adding $k - \varrho_{D_0}(v)$ parallel arcs from $s$ to $v$. Then, for any strict superset $X \subseteq V$ of $U_0$, the following are equivalent:
\begin{enumerate}[label=(\roman*)]
    \item $X$ violates the $(k,\ell)$-sparsity bound of $G$, i.e., $i_G(X) > k|X| - \ell$.
    \item $X \setminus U_0$ violates the rooted $(\ell - tk)$-arc-connectivity of $D_0'$, i.e., $\varrho_{D_0'}(X \setminus U_0) < \ell - tk$.
\end{enumerate}
\end{lemma}

\begin{proof}
Fix a vertex set $X \subseteq V$ such that $U_0 \subsetneq X$.
Let $q$ be the number of edges of $G$ with one endpoint in $U_0$ and the other in $X \setminus U_0$.
Then
\[
i_{D_0'}(X \setminus U_0) = i_G(X) - q,
\quad\text{and}\quad
\sum_{v \in X \setminus U_0} \varrho_{D_0'}(v) = k|X \setminus U_0| - q.
\]
Hence, we have
\[
\varrho_{D_0'}(X \setminus U_0) = \sum_{v \in X \setminus U_0} \varrho_{D_0'}(v) - i_{D_0'}(X \setminus U_0)
= k|X \setminus U_0| - i_G(X)
= k|X| - tk - i_G(X).
\]
Therefore,
\[
\varrho_{D_0'}(X \setminus U_0) < \ell - tk
\;\Longleftrightarrow\;
k|X| - tk - i_G(X) < \ell - tk
\;\Longleftrightarrow\;
i_G(X) > k|X| - \ell.
\qedhere
\]
\end{proof}

\begin{example}
Let $(k,\ell)=(2,3)$ and consider the graph $G$ on the left of Figure~\ref{fig:lemma_main} with $U_0=\{u\}$. The shaded four-vertex set $X$ induces $6>5=2\cdot4-3$ edges, violating the $(2,3)$-sparsity bound. The middle subfigure shows a $2$-indegree-bounded $U_0$-source orientation of $G$. From this, we construct the digraph $D_0'$ with root $s$ as prescribed.
As seen in the right subfigure, $X\setminus U_0$ violates the rooted $1$-arc-connectivity of $D_0'$.

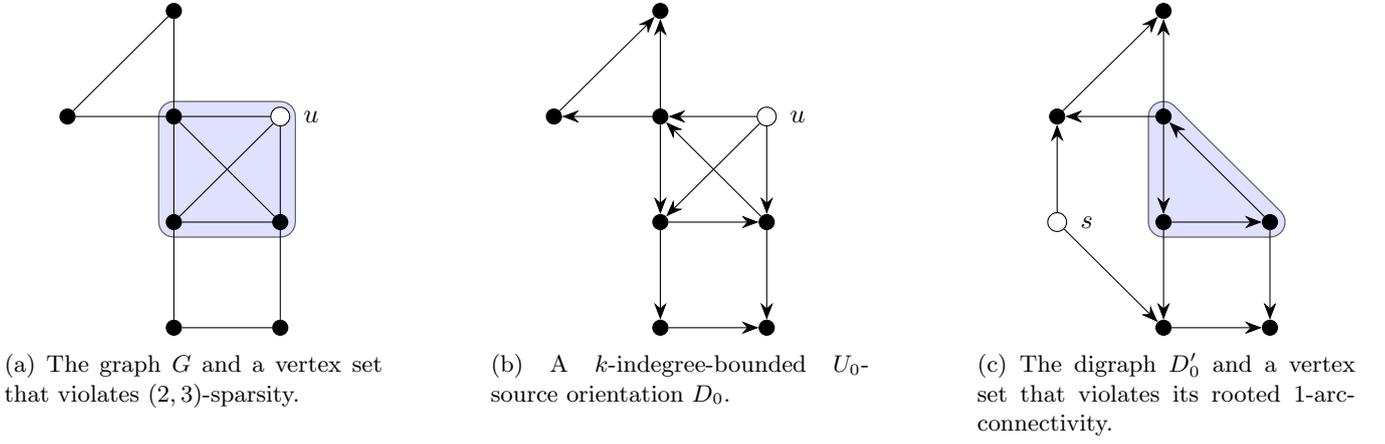
\begin{figure}[H]
\centering
  \begin{subfigure}[t]{0.28\textwidth}
    \centering
    \begin{tikzpicture}[scale=1.4]
      \node[draw, circle, fill=black, inner sep=2pt] (A) at (0,0) {};
      \node[draw, circle, fill=black, inner sep=2pt] (B) at (1,0) {};
      \node[draw, circle, fill=white, inner sep=2pt, label={[xshift=0.5mm]right:$u$}] (C) at (1,1) {};
      \node[draw, circle, fill=black, inner sep=2pt] (D) at (0,1) {};
      \node[draw, circle, fill=black, inner sep=2pt] (E) at (0,2) {};
      \node[draw, circle, fill=black, inner sep=2pt] (F) at (-1,1) {};
      \node[draw, circle, fill=black, inner sep=2pt] (G) at (0,-1) {};
      \node[draw, circle, fill=black, inner sep=2pt] (H) at (1,-1) {};

      \begin{pgfonlayer}{background}
        \begin{scope}[opacity=.6, transparency group]
          \path[line cap=round, line join=round, line width=0.5pt, double=blue!20, double distance=3.8mm, draw=black]{(A.center) -- (B.center) -- (C.center) -- (D.center) -- cycle};
          \highlight{3.7mm}{blue!20, fill=blue!20}{(A.center) to (B.center) to (C.center) to (D.center) -- cycle}
        \end{scope}
      \end{pgfonlayer}
      
      \draw (A) -- (B);
      \draw (A) -- (C);
      \draw (A) -- (D);
      \draw (B) -- (C);
      \draw (B) -- (D);
      \draw (C) -- (D);
      \draw (A) -- (G);
      \draw (G) -- (H);
      \draw (B) -- (H);
      \draw (D) -- (E);
      \draw (D) -- (F);
      \draw (E) -- (F);
    \end{tikzpicture}
    \caption{The graph $G$ and a vertex set that violates $(2, 3)$-sparsity.}
  \end{subfigure}
  \hfill
  \begin{subfigure}[t]{0.28\textwidth}
    \centering
    \begin{tikzpicture}[scale=1.4]
      \node[draw, circle, fill=black, inner sep=2pt] (A) at (0,0) {};
      \node[draw, circle, fill=black, inner sep=2pt] (B) at (1,0) {};
      \node[draw, circle, fill=white, inner sep=2pt, label={[xshift=0.5mm]right:$u$}] (C) at (1,1) {};
      \node[draw, circle, fill=black, inner sep=2pt] (D) at (0,1) {};
      \node[draw, circle, fill=black, inner sep=2pt] (E) at (0,2) {};
      \node[draw, circle, fill=black, inner sep=2pt] (F) at (-1,1) {};
      \node[draw, circle, fill=black, inner sep=2pt] (G) at (0,-1) {};
      \node[draw, circle, fill=black, inner sep=2pt] (H) at (1,-1) {};

      \draw[\niceArrow] (A) to (B);
      \draw[\niceArrow] (C) to (A);
      \draw[\niceArrow] (D) to (A);
      \draw[\niceArrow] (C) to (B);
      \draw[\niceArrow] (B) to (D);
      \draw[\niceArrow] (C) to (D);
      \draw[\niceArrow] (A) to (G);
      \draw[\niceArrow] (G) to (H);
      \draw[\niceArrow] (B) to (H);
      \draw[\niceArrow] (D) to (E);
      \draw[\niceArrow] (D) to (F);
      \draw[\niceArrow] (F) to (E);
    \end{tikzpicture}
    \caption{A $k$-indegree-bounded $U_0$-source orientation $D_0$.}
  \end{subfigure}
  \hfill
  \begin{subfigure}[t]{0.28\textwidth}
    \centering
    \begin{tikzpicture}[scale=1.4]
      \node[draw, circle, fill=black, inner sep=2pt] (A) at (0,0) {};
      \node[draw, circle, fill=black, inner sep=2pt] (B) at (1,0) {};
      \node[draw, circle, fill=black, inner sep=2pt] (D) at (0,1) {};
      \node[draw, circle, fill=black, inner sep=2pt] (E) at (0,2) {};
      \node[draw, circle, fill=black, inner sep=2pt] (F) at (-1,1) {};
      \node[draw, circle, fill=black, inner sep=2pt] (G) at (0,-1) {};
      \node[draw, circle, fill=black, inner sep=2pt] (H) at (1,-1) {};
      \node[draw, circle, fill=white, inner sep=2pt, label={[xshift=0.5mm]right:$s$}] (S) at (-1,0) {};

      \begin{pgfonlayer}{background}
        \begin{scope}[opacity=.6, transparency group]
          \path[line cap=round, line join=round, line width=0.5pt, double=blue!20, double distance=3.8mm, draw=black]{(A.center) -- (B.center) -- (D.center) -- cycle};
          \highlight{3.7mm}{blue!20, fill=blue!20}{(A.center) to (B.center) to (D.center) -- cycle}
        \end{scope}
      \end{pgfonlayer}

      \draw[\niceArrow] (A) to (B);
      \draw[\niceArrow] (D) to (A);
      \draw[\niceArrow] (B) to (D);
      \draw[\niceArrow] (A) to (G);
      \draw[\niceArrow] (G) to (H);
      \draw[\niceArrow] (B) to (H);
      \draw[\niceArrow] (D) to (E);
      \draw[\niceArrow] (D) to (F);
      \draw[\niceArrow] (F) to (E);
      \draw[\niceArrow] (S) to (F);
      \draw[\niceArrow] (S) to (G);
    \end{tikzpicture}
    \caption{The digraph $D_0'$ and a vertex set that violates its rooted $1$-arc-connectivity.}
  \end{subfigure}
  \caption{Application of Lemma~\ref{lemma:main}.}\label{fig:lemma_main}
\end{figure}

\end{example}

Let $U_0$ be an independent set of $G$, and suppose we already have a $k$-indegree-bounded $U_0$-source orientation. By Lemma~\ref{lemma:main}, the question of whether there exists a strict superset of $U_0$ that violates the $(k, \ell)$-sparsity of $G$ reduces to a rooted arc-connectivity problem. This yields the following algorithm:

\begin{algorithm}[H]
    \caption{Check Superset Sparsity} \label{alg:check_superset_sparsity}
    \begin{algorithmic}[1]
        \Input A $k$-indegree-bounded digraph $D_0=(V, A)$, and an independent vertex set $U_0 \subseteq V$ with $t$ elements.
        
    	\Require The inequalities $tk \leq \ell \leq (t+1)k$ are satisfied and $D_0$ is $U_0$-source.
        \Output A vertex set strictly containing $U_0$ that violates the $(k, \ell)$-sparsity of the underlying graph of $D_0$ if it exists; otherwise $\varnothing$.
        \vspace{0.2em}
        \hrule
        \vspace{0.2em}
        \Procedure{CheckSupsetSparsity$_{k,\ell}$}{$D_0=(V, A), U_0$}
        	\State $D_0' \gets (D_0 \setminus U_0) \cup \{s\}$ \Comment{Initialize the digraph $D_0'$ of Lemma~\ref{lemma:main} with root $s$}
        	\For{$v \in V \setminus U_0$}
        		\For{$i \gets 1, \dots, k - \varrho_{D_0}(v)$}
        			\State $D_0' \gets D_0' \cup \{sv\}$ \Comment{Insert an arc from $s$ to $v$ in $D_0'$}
        		\EndFor
        	\EndFor
        	\State $X \gets$ \Call{RootedArcConn$_{\ell-tk}$}{$D_0', s$} \Comment{Reduction to rooted arc-connectivity}
        	\If{$X = \varnothing$}
        		\State \Return $\varnothing$
        	\Else
        		\State \Return $X \cup U_0$
        	\EndIf
        \EndProcedure
    \end{algorithmic}
\end{algorithm}

\begin{proposition}\label{proposition:main}
Algorithm~\ref{alg:check_superset_sparsity} runs in $O(T_{RC}(n, \ell - tk))$ time.
\end{proposition}

\begin{proof}
We construct $D_0'$ in linear time. Since $D_0'$ has $O(n)$ vertices and $O(n)$ arcs, it takes $T_{RC}(n, \ell - tk)$ time to decide whether it is rooted $(\ell - tk)$-arc-connected. Therefore, the total running time of the algorithm~is
\[
  O(T_{RC}(n, \ell - tk) + n) = O(T_{RC}(n, \ell - tk)).\qedhere
\]
\end{proof}

To apply Lemma~\ref{lemma:main} and Algorithm~\ref{alg:check_superset_sparsity}, we require a $k$-indegree-bounded $U_0$-source orientation. We show that such an orientation --- or a certificate of its nonexistence --- can be obtained from an arbitrary $k$-indegree-bounded orientation using augmenting paths. Indeed, suppose that $D$ is a $k$-indegree-bounded orientation that is not $U_0$-source, and that there is no directed path in $D$ from
\[
S=\{v \in V \setminus U_0 : \varrho_D(v) < k\}
\]
to $U_0$. Let $T$ be the set of vertices from which at least one vertex of $U_0$ is reachable in $D$. Then $\varrho_D(T)=0$, and $\varrho_D(v)=k$ for every $v \in T \setminus U_0$. Hence,
\[
i_G(T) = \sum_{v \in T} \varrho_D(v) > \sum_{v \in T \setminus U_0} \varrho_D(v) = k|T \setminus U_0| = k|T| - tk.
\]
Clearly, $T \setminus U_0 \neq \emptyset$; therefore, by Lemma~\ref{lemma:u0_source}, no $k$-indegree-bounded $U_0$-source orientation exists.

This yields an augmenting-path algorithm that, starting from any $k$-indegree-bounded orientation $D$, finds either a suitable orientation $D_0$ or a certificate of its nonexistence. In each step, we find a path from $S$ to $U_0$ and reverse its arcs, thereby decreasing the sum of indegrees in $U_0$. If there is no path from $S$ to $U_0$, then either the current orientation is already suitable or no such orientation exists. We give the exact implementation below.

\begin{algorithm}[H]
    \caption{Reorient} \label{alg:transform_orientation}
    \begin{algorithmic}[1]
        \Input A $k$-indegree-bounded digraph $D=(V, A)$, and an independent vertex set $U_0 \subseteq V$ with $t$ elements.
        \Output A pair $(T, D_0)$, where $T=\varnothing$ and $D_0$ is a $k$-indegree-bounded $U_0$-source reorientation of $D$ if it exists; otherwise, $T \supsetneq U_0$ is a vertex set that violates the $(k, tk)$-sparsity of the underlying graph of $D$, and $D_0 = \varnothing$.
        \vspace{0.2em}
        \hrule
        \vspace{0.2em}
        \Procedure{Reorient$_{k}$}{$D=(V,A), U_0$}
            \While{$D$ is not $U_0$-source}
            	\State $P \gets$ a path in $D$ from $\{v \in V \setminus U_0: \varrho_D(v) < k\}$ to $U_0$
      	  		\If{no such path exists}
      	  			\State $T \gets$ set of vertices from which $U_0$ is reachable
      	    		\State \Return $(T, \varnothing)$ \Comment{No suitable orientation exists}
      	  		\EndIf
      	  		\State reverse the arcs of $P$ in $D$ \Comment{Decrease the sum of indegrees in $U_0$}
            \EndWhile
            \State \Return $(\varnothing, D)$
        \EndProcedure
    \end{algorithmic}
\end{algorithm}

\begin{proposition}\label{proposition:d0_orientation}
Algorithm~\ref{alg:transform_orientation} runs in $O(tn)$ time.
\end{proposition}

\begin{proof}
Since the original digraph $D$ is $k$-indegree-bounded and $|U_0|=t$, we have $\sum_{v\in U_0}\varrho_D(v) \leq |U_0| \cdot k=tk$. Reversing an augmenting path decreases the sum of the indegrees in $U_0$ by $1$, so the algorithm makes at most $tk$ iterations. Since each augmenting path, as well as the vertex set $T$, can be found in $O(n)$ time by a single graph traversal, the overall running time is $O(tn)$.
\end{proof}

\begin{note}
Throughout this paper, the number $t$ is treated as a constant whenever the procedure is invoked; this way, the time complexity of Algorithm~\ref{alg:transform_orientation} becomes $O(n)$.
\end{note}

\section{Recognition algorithms}
\label{sec:recognition_algorithms}

\subsection{The range $\ell \leq k$}

For $\ell\leq k$, the Gabow--Westermann algorithm~\cite{gabow1988forests} runs in $O(n\sqrt n)$ time when $\ell=0$ and in $O(n\sqrt{n\log n})$ time when $\ell>0$, while the boundary case $\ell=k$ also admits the near-linear forest-packing algorithm of Arkhipov and Kolmogorov~\cite{arkhipov2026faster}.
By Lemma~\ref{lemma:main} and Proposition~\ref{proposition:main}, we obtain a single solution for the entire range, as described below.

\begin{algorithm}[H]
    \caption{Check Sparsity for $\ell \leq k$} \label{alg:small_l}
    \begin{algorithmic}[1]
        \Input An undirected graph $G = (V, E)$.
        \Output A vertex set that violates the $(k, \ell)$-sparsity of $G$ if it exists; otherwise $\varnothing$.
        \vspace{0.2em}
        \hrule
        \vspace{0.2em}
        \Procedure{CheckSparsity$_{k,\ell}$}{$G=(V,E)$}
            \State $(X, D) \gets$ \Call{BoundedOri$_{k}$}{$G$}
            \If{$X \neq \varnothing$} \Comment{$X$ violates $(k, 0)$-sparsity}
                \State \Return $X$
            \EndIf
            \State \Return \Call{CheckSupsetSparsity$_{k,\ell}$}{$D$, $\emptyset$}
        \EndProcedure
    \end{algorithmic}
\end{algorithm}

\begin{theorem}
Algorithm~\ref{alg:small_l} correctly recognizes $(k,\ell)$-sparsity for $0 \leq \ell \leq k$.
If the input graph is not $(k,\ell)$-sparse, then the algorithm returns a vertex set that violates $(k,\ell)$-sparsity; otherwise, it returns $\varnothing$.
Its running time is $O(T_{BO}(n,k) + T_{RC}(n,\ell))$.
\end{theorem}

\begin{proof}
If \textsc{BoundedOri}$_k$ returns a vertex set $X$, then $X$ violates $(k,0)$-sparsity and hence also $(k,\ell)$-sparsity.
Otherwise, it returns a $k$-indegree-bounded orientation $D$.
By Lemma~\ref{lemma:main} with $t=0$ and $U_0=\varnothing$, \textsc{CheckSupsetSparsity}$_{k,\ell}$ returns a vertex set violating $(k,\ell)$-sparsity if one exists, and returns $\varnothing$ otherwise.
This proves correctness.
Computing a $k$-indegree-bounded orientation takes $O(T_{BO}(n,k))$ time.
Together with Proposition~\ref{proposition:main}, which gives a running time of $O(T_{RC}(n,\ell))$ for \textsc{CheckSupsetSparsity}, this yields the claimed bound.
\end{proof}

\subsection{The range $k < \ell < 2k$}

In the range $k < \ell < 2k$, Lemma~\ref{lemma:main} can be applied with $t=1$ to test efficiently whether there exists a violating vertex set containing a prescribed vertex. Repeating this test for every vertex gives an algorithm for deciding $(k,\ell)$-sparsity, but this naive approach still does not improve on the quadratic time bound. To obtain a faster algorithm, we use an additional structural property: by a theorem of Nash-Williams~\cite{nash1961edgeDisjointST}, every $(k,\ell)$-sparse graph with $k \leq \ell < 2k$ admits a decomposition into $k$ forests. Our algorithm for this parameter range exploits this decomposition. Throughout this section, all graphs are assumed to be loopless. We begin with a basic definition.

\begin{definition}
Given a forest $F=(V, E)$, we say that $F$ \emph{saturates} a nonempty vertex set $X \subseteq V$ if the induced subgraph $F[X]$ is connected, i.e., it contains exactly $|X|-1$ edges.
\end{definition}

\begin{lemma}
\label{lemma:violating_set}

Let $k < \ell < 2k$ and suppose the edge set of $G=(V,E)$ decomposes into $k$ edge-disjoint spanning forests $F_1,\dots,F_k$. Then every vertex set that violates the $(k,\ell)$-sparsity of $G$ is saturated by at least one of $F_1,\dots,F_{\ell-k}$.
\end{lemma}

\begin{proof}
Let $X \subseteq V$ be a nonempty vertex set that is not saturated by any of $F_1,\dots,F_{\ell-k}$.
Every forest saturates a singleton, so $|X|\geq2$.
For $i\leq\ell-k$, the forest $F_i[X]$ is disconnected and therefore has at most $|X|-2$ edges.
Thus
\[
i_G(X)=\sum_{i=1}^k i_{F_{i}}(X) \leq \sum_{i=1}^{\ell - k} (|X|-2) + \sum_{i=\ell - k + 1}^{k} (|X|-1) = k|X|-2(\ell - k)-(2k - \ell) = k|X| - \ell.
\]
Thus $X$ does not violate the $(k,\ell)$-sparsity condition.
\end{proof}

By the lemma above, we may restrict attention to violating sets that are saturated by $F_i$ for some $i = 1, \dots, \ell-k$. Since a forest saturates a vertex set $X$ if and only if one of its connected components does, it is enough to consider the connected components of $F_1, \dots, F_{\ell-k}$. This reduces our problem to the following: given a graph $G$ and a spanning tree $T$ of $G$, either find a violating vertex set in $G$ or correctly conclude that no violating vertex set saturated by $T$ exists. Our first goal is to design a procedure for this task. We begin with a simple observation.

\begin{observation}
Let $T = (V,E)$ be a tree, and let $X \subseteq V$ be a vertex set saturated by $T$. Fix any vertex $c \in V$, and let $T_1, \dots, T_q$ denote the connected components of $T \setminus \{c\}$. Then either $c \in X$ or there exists an index $i$ such that $X \subseteq V(T_i)$ and $X$ is saturated by $T_i$.
\end{observation}

\begin{proof}
By definition, a vertex set $X$ saturated by $T$ induces a subtree of $T$. Hence, either $c \in X$ or $X \subseteq V(T_i)$ for some index $i$. In the latter case, every edge of $T$ induced by $X$ belongs to $T_i$, and so $X$ is saturated by $T_i$ as well.
\end{proof}

The observation above naturally leads to a divide-and-conquer approach based on centroid decomposition~\cite{jordan1869assemblages}. Suppose we are given a $k$-indegree-bounded orientation of $G$ together with a spanning tree $T$. By Propositions~\ref{proposition:main} and~\ref{proposition:d0_orientation}, we can efficiently determine whether there exists a violating vertex set containing the centroid $c$ of $T$, and find one if it does. If so, we are done. Otherwise, by the observation above, it remains to recurse on the connected components of $T \setminus \{c\}$. The full implementation is given below.

\begin{algorithm}[H]
  \caption{Saturated Violation}\label{alg:csvs}
  \begin{algorithmic}[1]
    \Input A $k$-indegree-bounded digraph $D=(V, A)$, and a spanning tree $T$ of the underlying graph of $D$.
   	\Output Either $\varnothing$ or a vertex set violating $(k,\ell)$-sparsity in the underlying graph of $D$. The procedure returns a violating set whenever the underlying graph of $D$ contains one that is saturated by $T$.
    \vspace{0.2em}
    \hrule
    \vspace{0.2em}
    \Procedure{SaturatedViolation$_{k,\ell}$}{$D=(V,A), T$}
		\State $c \gets$ the centroid of $T$\label{alg:csvs:find_centroid}
		\State $(X, D_0) \gets$ \Call{Reorient$_{k}$}{$D$, $\{c\}$}
		\If{$X \neq \varnothing$} \Comment{$X \supsetneq \{c\}$ violates $(k, k)$-sparsity}
			\State \Return $X$
		\EndIf
		\State $X \gets$ \Call{CheckSupsetSparsity$_{k,\ell}$}{$D_0$, $\{c\}$}
		\If {$X \neq \varnothing$} \Comment{$X \supsetneq \{c\}$ violates $(k, \ell)$-sparsity}
			\State \Return $X$
		\EndIf
		\State $T_1, \dots, T_q \gets$ connected components of $(T \setminus \{c\})$
		\For{$i \gets 1, \dots, q$}
			\State $X \gets$ \Call{SaturatedViolation$_{k,\ell}$}{$D\left[V(T_i)\right], T_i$}
			\If{$X \neq \varnothing$}
				\State \Return $X$
			\EndIf
		\EndFor
		\State \Return $\varnothing$
    \EndProcedure
  \end{algorithmic}
\end{algorithm}

\paragraph{Implementation details for \textnormal{\textsc{SaturatedViolation}}.} We find the centroid $c$ using two depth-first searches: the first computes, for each vertex $v$, the sizes of the subtrees incident to $v$, and the second locates a vertex whose largest such subtree has at most half of the vertices of $T$. The components $T_1,\dots,T_q$ with vertex sets $V_1, \dots, V_q$ are obtained by $q$ graph traversals. For each $i$, we build the induced digraph $D[V_i]$ by scanning the vertices of $V_i$, collecting their incoming arcs, and keeping only those whose tail also lies in~$V_i$.

\begin{note}
If the underlying graph of $D$ contains a violating vertex set, but no such set is saturated by $T$, then \textsc{SaturatedViolation} may return either $\varnothing$ or a violating set. This does not affect the correctness of the final recognition algorithm.
\end{note}

\tikzset{
  wavy/.style={
    decorate,
    decoration={snake, amplitude=0.22mm, segment length=2mm, post length=2.5mm},
    draw=black!80,
    line cap=round, line join=round,
    -{Stealth[length=2.4mm,width=2mm,fill=black]},
  },
  thick/.style={
    \niceArrow,
    line width=0.8pt,
    -{Stealth[length=2.4mm,width=2mm]},
  },
}

\begin{example}

Let $(k,\ell)=(2,3)$, and consider the $2$-indegree-bounded digraph $D$ shown at the top of Figure~\ref{fig:centroid}, together with a spanning tree $T$ whose edges are drawn as straight segments. First, find a centroid $c$ of $T$. Since no violating vertex set contains $c$, \textsc{SaturatedViolation} recurses on the two subtrees of $T$ incident to $c$ and solves the corresponding subproblems independently.

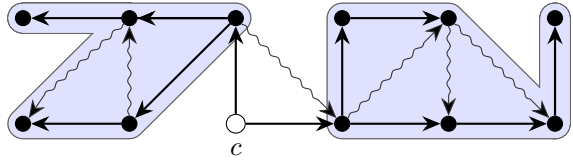
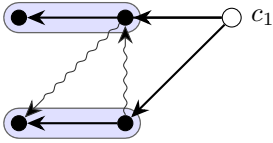
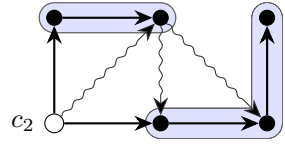
\begin{figure}[H]
\centering
\begin{subfigure}[t]{0.8\textwidth}
\centering
\begin{tikzpicture}[scale=1.4]
\begin{scope}
	\node[draw, circle, fill=black, inner sep=2pt] (K) at (-2,1) {};
	\node[draw, circle, fill=black, inner sep=2pt] (A) at (-2,0) {};
	\node[draw, circle, fill=black, inner sep=2pt] (B) at (-1,1) {};
	\node[draw, circle, fill=black, inner sep=2pt] (C) at (-1,0) {};
	\node[draw, circle, fill=black, inner sep=2pt] (D) at (0,1) {};
	\node[draw, circle, fill=white, inner sep=2pt, label={[yshift=-6.75mm]$c$}] (E) at (0,0) {};
	\node[draw, circle, fill=black, inner sep=2pt] (F) at (1,1) {};
	\node[draw, circle, fill=black, inner sep=2pt] (G) at (1,0) {};
	\node[draw, circle, fill=black, inner sep=2pt] (H) at (2,1) {};
	\node[draw, circle, fill=black, inner sep=2pt] (I) at (2,0) {};
	\node[draw, circle, fill=black, inner sep=2pt] (J) at (3,0) {};
	\node[draw, circle, fill=black, inner sep=2pt] (L) at (3,1) {};
	
	\begin{pgfonlayer}{background}
		\begin{scope}[opacity=.6, transparency group]
		\path[line cap=round, line join=round, line width=0.5pt, double=blue!20, double distance=3.8mm, draw=black]{(A.center) -- (C.center) -- (D.center) -- (K.center) -- (B.center) --cycle};
		\highlight{3.7mm}{blue!20, fill=blue!20}{(A.center) to (C.center) to (D.center) to (K.center) -- (B.center) -- cycle}
        \end{scope}
	\end{pgfonlayer}
      
	\begin{pgfonlayer}{background}
		\begin{scope}[opacity=.6, transparency group]
		\path[line cap=round, line join=round, line width=0.5pt, double=blue!20, double distance=3.8mm, draw=black]{(L.center) -- (J.center) -- (G.center) -- (F.center) -- (H.center) -- (J.center) -- cycle};
		\highlight{3.7mm}{blue!20, fill=blue!20}{(L.center) to (J.center) to (G.center) to (F.center) to (H.center) to (J.center) -- cycle}
        \end{scope}
	\end{pgfonlayer}

	\draw[wavy] (B) to (A);
	\draw[thick] (C) to (A);
	\draw[wavy] (C) to (B);
	\draw[thick] (D) to (C);
	\draw[thick] (D) to (B);
	\draw[thick] (B) to (K);
	\draw[thick] (E) to (D);
	\draw[thick] (E) to (G);
	\draw[wavy] (D) to (G);
	\draw[wavy] (G) to (H);
	\draw[wavy] (H) to (I);
	\draw[thick] (G) to (F);
	\draw[thick] (F) to (H);
	\draw[thick] (G) to (I);
	\draw[thick] (I) to (J);
	\draw[thick] (J) to (L);
	\draw[wavy] (H) to (J);

\end{scope}
\end{tikzpicture}
\caption{The digraph $D$, the spanning tree $T$, and a centroid $c$ of $T$.}
\end{subfigure}

\vspace{0.4cm}

\begin{subfigure}[t]{0.38\textwidth}
\centering
\begin{tikzpicture}[scale=1.4]
\begin{scope}[shift={(-0.7,-1.9)}]
	\node[draw, circle, fill=black, inner sep=2pt] (K) at (-2,1) {};
	\node[draw, circle, fill=black, inner sep=2pt] (A) at (-2,0) {};
	\node[draw, circle, fill=black, inner sep=2pt] (B) at (-1,1) {};
	\node[draw, circle, fill=black, inner sep=2pt] (C) at (-1,0) {};
	\node[draw, circle, fill=white, inner sep=2pt,label={right:$c_1$}] (D) at (0,1) {};

	\draw[wavy] (B) to (A);
	\draw[thick] (C) to (A);
	\draw[wavy] (C) to (B);
	\draw[thick] (D) to (C);
	\draw[thick] (D) to (B);
	\draw[thick] (B) to (K);

    \begin{pgfonlayer}{background}
		\begin{scope}[opacity=.6, transparency group]
		\path[line cap=round, line join=round, line width=0.5pt, double=blue!20, double distance=3.8mm, draw=black]{(B.center) -- (K.center) -- cycle};
		\highlight{3.7mm}{blue!20, fill=blue!20}{(B.center) to (K.center) -- cycle}
        \end{scope}
	\end{pgfonlayer}
    
    \begin{pgfonlayer}{background}
		\begin{scope}[opacity=.6, transparency group]
		\path[line cap=round, line join=round, line width=0.5pt, double=blue!20, double distance=3.8mm, draw=black]{(A.center) -- (C.center) -- cycle};
		\highlight{3.7mm}{blue!20, fill=blue!20}{(A.center) to (C.center) -- cycle}
        \end{scope}
	\end{pgfonlayer}
\end{scope}
\end{tikzpicture}
\caption{The subgraph induced by the left subtree of $c$ and a centroid $c_1$ of this subtree.}
\end{subfigure}
\hspace{1cm}
\begin{subfigure}[t]{0.38\textwidth}
\centering
\begin{tikzpicture}[scale=1.4]
\begin{scope}[shift={(-0.3,-1.9)}]
	\node[draw, circle, fill=black, inner sep=2pt] (F) at (1,1) {};
	\node[draw, circle, fill=white, inner sep=2pt,label={left:$c_2$}] (G) at (1,0) {};
	\node[draw, circle, fill=black, inner sep=2pt] (H) at (2,1) {};
	\node[draw, circle, fill=black, inner sep=2pt] (I) at (2,0) {};
	\node[draw, circle, fill=black, inner sep=2pt] (J) at (3,0) {};
	\node[draw, circle, fill=black, inner sep=2pt] (L) at (3,1) {};

	\draw[wavy] (G) to (H);
	\draw[wavy] (H) to (I);
	\draw[thick] (G) to (F);
	\draw[thick] (F) to (H);
	\draw[thick] (G) to (I);
	\draw[thick] (I) to (J);
	\draw[thick] (J) to (L);
	\draw[wavy] (H) to (J);
    
    \begin{pgfonlayer}{background}
		\begin{scope}[opacity=.6, transparency group]
		\path[line cap=round, line join=round, line width=0.5pt, double=blue!20, double distance=3.8mm, draw=black]{(F.center) -- (H.center) -- cycle};
		\highlight{3.7mm}{blue!20, fill=blue!20}{(F.center) to (H.center) -- cycle}
        \end{scope}
	\end{pgfonlayer}
    
    \begin{pgfonlayer}{background}
		\begin{scope}[opacity=.6, transparency group]
		\path[line cap=round, line join=round, line width=0.5pt, double=blue!20, double distance=3.8mm, draw=black]{(I.center) -- (J.center) -- (L.center) -- (J.center) -- cycle};
		\highlight{3.7mm}{blue!20, fill=blue!20}{(I.center) -- (J.center) -- (L.center) -- (J.center) -- cycle}
        \end{scope}
	\end{pgfonlayer}
\end{scope}

\end{tikzpicture}
\caption{The subgraph induced by the right subtree of $c$ and a centroid $c_2$ of this subtree.}
\end{subfigure}
\caption{Two steps of the centroid decomposition. The blue regions show the vertex sets of the connected components obtained after deleting the displayed centroid; these components define the recursive subproblems.}\label{fig:centroid}
\end{figure}

\end{example}

For completeness, we record a simple fact that we will use repeatedly in the analysis.

\begin{lemma}
\label{lemma:sum_time}
Let $f(n) = n \cdot g(n)$, where $g:\mathbb{Z}_{>0} \to \mathbb{R}_{>0}$ is nondecreasing, and suppose that $T(n) = O(f(n))$. Let $q$ be a positive integer. Then, for any positive integers $n_1, \dots, n_q$ with sum $N$,
we have
\[
\sum_{i=1}^q T(n_i) = O(f(N)).
\]
\end{lemma}

\begin{proof}
For a sufficiently large positive constant $C$, the inequality $T(n) \leq C \cdot f(n)$ holds for all $n \in \mathbb{Z}_{>0}$. Then, using the monotonicity of $g$, we have
\[
  \sum_{i=1}^q T(n_i) \leq C \cdot \sum_{i=1}^q f(n_i) = C \cdot \sum_{i=1}^q n_i \cdot g(n_i)\leq C \cdot \sum_{i=1}^q n_i \cdot g(N) = C \cdot N \cdot g(N) = C \cdot f(N).\qedhere
\]
\end{proof}

By the ``conquer'' phase of Algorithm~\ref{alg:csvs}, we mean all steps excluding the recursive calls. We first analyze the running time of this phase and then proceed to the analysis of the entire algorithm.

\begin{proposition}
\label{proposition:CSVS_time_without_recursion}
The conquer phase of Algorithm~\ref{alg:csvs} runs in $O(T_{RC}(n,\ell-k))$ time.
\end{proposition}

\begin{proof}
The two depth-first searches used to find a centroid each take $O(n)$ time. By Propositions~\ref{proposition:main} and~\ref{proposition:d0_orientation}, the calls to \textsc{CheckSupsetSparsity} and \textsc{Reorient} run in $O(T_{RC}(n,\ell-k))$ and $O(n)$ time, respectively. The subtrees $T_1,\dots,T_q$ with vertex sets $V_1,\dots,V_q$ are obtained by $q$ graph traversals, where the $i$th traversal takes $O(|V_i|)$ time. Hence, by Lemma~\ref{lemma:sum_time}, their total running time is
\[
O\left(\sum_{i=1}^q |V_i|\right)=O(n-1)=O(n).
\]
Moreover, since $D$ is $k$-indegree-bounded, constructing $D[V_i]$ requires checking at most $k$ incoming arcs for each vertex in $V_i$. Hence, the construction takes $O(|V_i|)$ time. Applying Lemma~\ref{lemma:sum_time} once more, we obtain a total running time of $O(n)$ for constructing all induced digraphs $D[V_i]$. Combining these bounds, the overall running time of the conquer phase is $O(T_{RC}(n,\ell-k)+n)$.
Since $\ell-k\geq 1$, the note following Table~\ref{tab:parameters} gives $T_{RC}(n,\ell-k)=\Omega(n)$, so this is $O(T_{RC}(n,\ell-k))$.
\end{proof}

\begin{proposition}
\label{proposition:CSVS_time}
Algorithm~\ref{alg:csvs} runs in $O(T'_{RC}(n, \ell - k) \log n)$ time.
\end{proposition}

\begin{proof}
Consider all recursive calls at a fixed recursion level. Let $D_1,\dots,D_p$ be the digraphs passed to the procedure at this level, and set $n_i = |V(D_i)|$ for each $i$. By Proposition~\ref{proposition:CSVS_time_without_recursion}, the conquer phase on $D_i$ runs in time
\[
O(T_{RC}(n_i,\ell-k)) \subseteq O(T'_{RC}(n_i,\ell-k)).
\]
Since the digraphs $D_1,\dots,D_p$ are pairwise vertex-disjoint subgraphs of the original digraph $D$, we have
\[
\sum_{i=1}^p n_i \leq |V(D)| = n.
\]
Applying Lemma~\ref{lemma:sum_time} with $N=\sum_{i=1}^p n_i\leq n$, the total running time spent on the conquer phases at this recursion level is $O(T'_{RC}(n,\ell-k))$.
Since the centroid chosen at line~\ref{alg:csvs:find_centroid} splits $T$ into subtrees of size at most $|V(T)|/2$, the recursion depth is $O(\log n)$.
Multiplying this by the $O(T'_{RC}(n,\ell-k))$ work per level gives the total running time
\[
O(T'_{RC}(n,\ell-k)\log n).\qedhere
\]
\end{proof}

We now return to the original problem of deciding whether a given graph $G$ is $(k,\ell)$-sparse.
We first compute a forest decomposition. Then, for each component of the first $\ell - k$ forests, we invoke \textsc{SaturatedViolation} to test whether $G$ contains a violating vertex set saturated by that component. The full implementation is as follows.

\begin{algorithm}[H]
  \caption{Check Sparsity for $k < \ell < 2k$} \label{alg:middle_l}
  \begin{algorithmic}[1]
    \Input An undirected graph $G = (V, E)$.
    \Output A vertex set that violates the $(k, \ell)$-sparsity of $G$ if it exists; otherwise $\varnothing$.
    \vspace{0.2em}
    \hrule
    \vspace{0.2em}
    \Procedure{CheckSparsity$_{k,\ell}$}{$G=(V,E)$}
		\State $(X, F_1, \dots, F_k) \gets$ \Call{ForestDecomp$_{k}$}{$G$}
		\If{$X \neq \varnothing$} \Comment{$X$ violates $(k, k)$-sparsity}
			\State \Return $X$
		\EndIf
		\State \Return \Call{CheckSparsityWithDecomp$_{k,\ell}$}{$G$, $F_1, \dots, F_k$}
	\EndProcedure
  \end{algorithmic}
\end{algorithm}

\begin{algorithm}[H]
  \caption{Check Sparsity for $k < \ell < 2k$ Given a Forest Decomposition} \label{alg:middle_l_decomp}
  \begin{algorithmic}[1]
    \Input An undirected graph $G = (V, E)$ and a forest decomposition $F_1, \dots, F_k$ of $G$.
    \Output A vertex set that violates the $(k, \ell)$-sparsity of $G$ if it exists; otherwise $\varnothing$.
    \vspace{0.2em}
    \hrule
    \vspace{0.2em}
    \Procedure{CheckSparsityWithDecomp$_{k,\ell}$}{$G=(V,E), F_1, \dots, F_k$}
		\State $D \gets$ a $k$-indegree-bounded orientation of $G$ \Comment{Use $F_1, \dots, F_k$ to obtain $D$} \label{alg:middle_l_decomp:bounded_ori}
		\For{$i \gets 1, \dots, \ell-k$}
			\State $T_{i, 1}, \dots, T_{i, q_i} \gets$ connected components of $F_i$
			\For{$j \gets 1, \dots, q_i$}
				\State $X \gets$ \Call{SaturatedViolation$_{k,\ell}$}{$D\left[V(T_{i,j})\right]$, $T_{i,j}$}
				\If{$X \neq \varnothing$}
					\State \Return $X$
				\EndIf
			\EndFor
		\EndFor
		\State \Return $\varnothing$
	\EndProcedure
  \end{algorithmic}
\end{algorithm}

\paragraph{Implementation details for \textnormal{\textsc{CheckSparsityWithDecomp}}.}
Since a forest decomposition is known, the $k$-indegree-bounded orientation $D$ at line~\ref{alg:middle_l_decomp:bounded_ori} can be obtained in linear time: choose a root in each tree and orient every edge from parent to child. With this orientation in hand, we find the components $T_{i,j}$ and build the induced subgraphs $D[V(T_{i,j})]$ exactly as in Algorithm~\ref{alg:csvs}.

\begin{proposition}
\label{proposition:middle_l_decomp}
Given a forest decomposition $F_1,\dots,F_k$ of $G$, Algorithm~\ref{alg:middle_l_decomp} correctly recognizes $(k,\ell)$-sparsity for $k<\ell<2k$.
If $G$ is not $(k,\ell)$-sparse, then the algorithm returns a vertex set that violates $(k,\ell)$-sparsity; otherwise, it returns $\varnothing$.
Its running time is $O(T'_{RC}(n,\ell-k)\log n)$.
\end{proposition}

\begin{proof}
By Lemma~\ref{lemma:violating_set}, every vertex set that violates $(k,\ell)$-sparsity is saturated by at least one of $F_1,\dots,F_{\ell-k}$, and hence by the connected component of that forest containing it.
For each such component, Algorithm~\ref{alg:csvs} returns a violating vertex set whenever one saturated by the component exists.
Consequently, Algorithm~\ref{alg:middle_l_decomp} returns a violating vertex set if one exists and returns $\varnothing$ otherwise.

The orientation $D$ is obtained in linear time.
Each component $T_{i,j}$, with $n_{i,j}$ vertices, can be identified in $O(n_{i,j})$ time by a single graph traversal.
For a fixed $i$, we have $\sum_{j=1}^{q_i} n_{i,j}=n$, so the total time spent over all components of the first $\ell-k$ forests is $O(n)$.
For each pair $(i,j)$, constructing the induced subgraph $D[V(T_{i,j})]$ takes $O(n_{i,j})$ time, since we check at most $k$ incoming arcs for each vertex.
By Proposition~\ref{proposition:CSVS_time}, the corresponding call to \textsc{SaturatedViolation} runs in $O(T'_{RC}(n_{i,j},\ell-k)\log n_{i,j})$ time.
Using $\log n_{i,j}\leq\log n$ and applying Lemma~\ref{lemma:sum_time} to $T'_{RC}$ gives a total running time of $O(T'_{RC}(n,\ell-k)\log n)$ for all calls to \textsc{SaturatedViolation}.
Since $\ell-k\geq1$, the note following Table~\ref{tab:parameters} gives $T'_{RC}(n,\ell-k)=\Omega(n)$.
Thus the total running time is $O(T'_{RC}(n,\ell-k)\log n)$.
\end{proof}

\begin{theorem}
\label{thm:middle_l}
Algorithm~\ref{alg:middle_l} correctly recognizes $(k,\ell)$-sparsity for $k < \ell < 2k$ on loopless graphs.
If the input graph is not $(k,\ell)$-sparse, then the algorithm returns a vertex set that violates $(k,\ell)$-sparsity; otherwise, it returns $\varnothing$.
Its running time is $O(T_{FD}(n,k) + T'_{RC}(n, \ell - k) \log n)$.
\end{theorem}

\begin{proof}
If \textsc{ForestDecomp}$_k$ returns a vertex set $X$, then $X$ violates $(k,k)$-sparsity and hence also $(k,\ell)$-sparsity.
Otherwise, it returns a forest decomposition $F_1,\dots,F_k$ of $G$.
By Proposition~\ref{proposition:middle_l_decomp}, the remaining steps take $O(T'_{RC}(n,\ell-k)\log n)$ time.
Together with the $O(T_{FD}(n,k))$ time needed to compute the forest decomposition, this gives the stated running time.
\end{proof}

\begin{corollary}
Under purely combinatorial implementations, Algorithm~\ref{alg:middle_l} runs in $O(n\sqrt n)$ time.
\end{corollary}

\begin{proof}
The forest decomposition can be computed in $O(n\sqrt n)$ time using purely combinatorial subroutines, as noted in Table~\ref{tab:parameters}. Moreover, $T'_{RC}(n,\ell-k)=O(n\log n)$, so Proposition~\ref{proposition:middle_l_decomp} gives an $O(n\log^2 n)$ bound for the remaining steps.
Since $n\log^2 n=O(n\sqrt n)$, the claim follows from Theorem~\ref{thm:middle_l}.
\end{proof}

\subsection{The range $2k \leq \ell < 3k$}
In the range $2k \leq \ell < 3k$, we follow the standard convention that the sparsity condition is required only for vertex sets of size at least three, and we restrict attention throughout to simple graphs.

\subsubsection{The range $\ell = 2k$}

To handle the case $\ell=2k$, we make a few modifications to Algorithm~\ref{alg:middle_l}; more precisely, to the subroutine \textsc{SaturatedViolation}. First, observe that every $(k,2k)$-sparse graph is also $(k,k)$-sparse, and therefore admits a decomposition into $k$ forests. Second, note that Lemma~\ref{lemma:violating_set} remains valid for $\ell=2k$.

\begin{lemma}
\label{lemma:violating_set_2k}
Suppose the edge set of $G=(V,E)$ decomposes into $k$ edge-disjoint spanning forests $F_1,\dots,F_k$. Then every vertex set that violates the $(k,2k)$-sparsity of $G$ is saturated by at least one of $F_1,\dots,F_k$.
\end{lemma}

\begin{proof}
The proof is analogous to that of Lemma~\ref{lemma:violating_set}.
\end{proof}

Although the lemma above extends unchanged to the case $\ell=2k$, we still cannot apply \textsc{SaturatedViolation} directly.
We call a set $X$ a \emph{weak violating set} if $i_G(X)>k|X|-2k$, without imposing the condition $|X|\geq3$.
A weak violating set violates $(k,2k)$-sparsity exactly when it has at least three vertices.
In Algorithm~\ref{alg:middle_l}, after finding the centroid $c$, we determine whether there exists a violating set containing $c$.
We would like to perform the same test when $\ell=2k$.
The difficulty is that \textsc{CheckSupsetSparsity} may return a weak violating set of size $2$, which does not violate $(k,2k)$-sparsity.
Any such set that strictly contains $c$ consists of $c$ and one of its neighbours, because the weak-violation inequality for a two-vertex set requires an edge between its vertices.
Consequently, if a weak violating superset $X$ of $c$ contains at least one vertex outside $N_G(c) \cup \{c\}$, then $|X|\geq3$, and hence $X$ violates $(k,2k)$-sparsity.

Suppose now that we are given a $k$-indegree-bounded orientation $D_0$ in which $c$ is a source.
We first need a subroutine that, for a given set $Z$ of vertices, determines whether there exists a weak violating set containing $c$ and at least one vertex of $Z$.
Applying Lemma~\ref{lemma:main} with $t=1$ yields the following algorithm.

\begin{algorithm}[H]
    \caption{Check Terminal Superset Sparsity}
\label{alg:check_term_superset_sparsity}
    \begin{algorithmic}[1]
        \Input A $k$-indegree-bounded digraph $D_0=(V, A)$, a vertex $c \in V$, and a set $Z \subseteq V \setminus \{c\}$ of vertices.
        \Require $c$ is a source in $D_0$.
        \Output A weak violating set $X$ of the underlying graph of $D_0$ such that $c\in X$, $X\cap Z\neq\emptyset$, and $X\neq\{c\}$, if such a set exists; otherwise $\varnothing$.
        \vspace{0.2em}
        \hrule
        \vspace{0.2em}
        \Procedure{CheckTermSupsetSparsity$_{k}$}{$D_0=(V, A), c, Z$}
			\State $D_0' \gets (D_0 \setminus \{c\}) \cup \{s\}$ \Comment{Initialize the auxiliary digraph of Lemma~\ref{lemma:main} for $t=1$, with root $s$}
			\For{$v \in V \setminus \{c\}$}
				\For{$i \gets 1, \dots, k - \varrho_{D_0}(v)$}
					\State $D_0' \gets D_0' \cup \{sv\}$ \Comment{Insert an arc from $s$ to $v$}
				\EndFor
			\EndFor
			\State $X \gets$ \Call{RootedTermConn$_{k}$}{$D_0', s, Z$} \Comment{Reduce to rooted terminal arc-connectivity}
			\If{$X = \varnothing$}
				\State \Return $\varnothing$
			\Else
				\State \Return $X \cup \{c\}$
			\EndIf
        \EndProcedure
    \end{algorithmic}
\end{algorithm}

\begin{proposition}
\label{proposition:check_term_superset_sparsity}
Algorithm~\ref{alg:check_term_superset_sparsity} is correct and runs in $O(T_{RTC}(n, k))$ time.
\end{proposition}
\begin{proof}
Let $G$ be the underlying graph of $D_0$.
By Lemma~\ref{lemma:main} with $t=1$, $\ell=2k$, and $U_0=\{c\}$, a set $Y\supsetneq\{c\}$ is a weak violating set if and only if $\varrho_{D_0'}(Y\setminus\{c\})<k$.
Since $Z\subseteq V\setminus\{c\}$, we have $Y\cap Z=(Y\setminus\{c\})\cap Z$, so the call to \textsc{RootedTermConn}$_k$ returns a set exactly when such a weak violating set $Y$ exists.
We construct $D_0'$ in linear time. Since $D_0'$ has $O(n)$ vertices and $O(n)$ arcs, the total running time is
\[
  O(T_{RTC}(n, k) + n) = O(T_{RTC}(n, k)).\qedhere
\]
\end{proof}

Suppose that, for a given vertex $c$, we have already ruled out every violating set that contains $c$ and a vertex outside $N_G(c) \cup \{c\}$.
It remains to test whether there exists a violating set $X$ with $c\in X\subseteq N_G(c) \cup \{c\}$.
The following lemma gives an equivalent condition.

\begin{lemma}
\label{lemma:2k_reduction}
Let $\ell=2k$ with $k \geq 3$, let $G=(V, E)$ be a simple graph, and let $c \in V$. Let $X \subseteq N_G(c)$ with $|X| \geq 2$. Then $X \cup \{c\}$ violates $(k,2k)$-sparsity in $G$ if and only if $X$ violates $(k-1,k)$-sparsity in $G[N_G(c)]$.
\end{lemma}

\begin{proof}
By the simplicity of $G$, we have $i_G(X \cup \{c\}) = i_G(X) + |X|$, while $i_G(X)=i_{G[N_G(c)]}(X)$ because $X\subseteq N_G(c)$.
Hence,
\[
i_{G[N_G(c)]}(X) > (k-1)|X| - k\; \Longleftrightarrow\; i_G(X \cup \{c\}) = i_G(X) + |X| > k|X \cup \{c\}| - 2k.
\]
\end{proof}

\begin{note}
For $k \leq 2$, the existence of a violating set $X$ with $c\in X\subseteq N_G(c) \cup \{c\}$ can be checked directly.
Such a set can exist only if $|N_G(c)| \geq 2$.
\begin{itemize}
\item If $k=1$, then $N_G(c) \cup \{c\}$ is a violating set.
\item If $k=2$, then for distinct vertices $u,v\in N_G(c)$, the set $\{u,v,c\}$ violates $(2,4)$-sparsity if and only if $uv\in E$. Hence, such a set exists if and only if $G[N_G(c)]$ contains an edge.
\end{itemize}
\end{note}

For $k \geq 3$, Lemma~\ref{lemma:2k_reduction} reduces the search for a violating set $X$ with $c\in X\subseteq N_G(c) \cup \{c\}$ to recognizing $(k-1,k)$-sparsity in $G[N_G(c)]$.
Since $k-1<k<2(k-1)$, the reduced $(k-1,k)$-sparsity problem satisfies the parameter condition of Proposition~\ref{proposition:middle_l_decomp}.
After \textsc{Reorient}$_k$ succeeds, however, we do not need to compute a forest decomposition of $G[N_G(c)]$ from scratch.
In the resulting orientation $D_0$, the vertex $c$ is a source, so every $v\in N_G(c)$ has an incoming arc from $c$ and $D_0[N_G(c)]$ is $(k-1)$-indegree-bounded.
We can therefore apply \textsc{BoundedForestDecomp}$_{k-1}$ and, if it returns a forest decomposition, invoke \textsc{CheckSparsityWithDecomp}$_{k-1,k}$.
This gives the following modified version of Algorithm~\ref{alg:csvs}.

\begin{algorithm}[H]
  \caption{Saturated Violation for $\ell=2k$}\label{alg:csvs_2k}
  \begin{algorithmic}[1]
    \Input A $k$-indegree-bounded digraph $D=(V, A)$ with simple underlying graph, and a spanning tree $T$ of the underlying graph of $D$.
   	\Output Either $\varnothing$ or a vertex set violating $(k,2k)$-sparsity in the underlying graph of $D$. The procedure returns a violating set whenever the underlying graph of $D$ contains one that is saturated by $T$.
    \vspace{0.2em}
    \hrule
    \vspace{0.2em}
    \Procedure{SaturatedViolation$_{k,2k}$}{$D=(V,A), T$}
		\State $G \gets$ the underlying graph of $D$
		\State $c \gets$ the centroid of $T$\label{alg:csvs_2k:find_centroid}
		\State $(X, D_0) \gets$ \Call{Reorient$_{k}$}{$D$, $\{c\}$}
		\If{$X \neq \varnothing$} \Comment{$X \supsetneq \{c\}$ violates $(k, k)$-sparsity}
			\State \Return $X$
		\EndIf
		
		\State $X \gets$ \Call{CheckTermSupsetSparsity$_{k}$}{$D_0$, $c$, $V\setminus(N_G(c)\cup\{c\})$}
		
		\If {$X \neq \varnothing$} \Comment{$X \supsetneq \{c\}$ violates $(k,2k)$-sparsity}
			\State \Return $X$
		\EndIf
		
		\If {$k=1$ \AND $|N_G(c)| \geq 2$}
			\State \Return $N_G(c) \cup \{c\}$ \Comment{$N_G(c) \cup \{c\}$ violates $(k,2k)$-sparsity}
		\EndIf
		
		\If {$k=2$ \AND $i_{G}(N_G(c)) > 0$}
			\State $uv \gets$ any edge of $G[N_G(c)]$
			\State \Return $\{u, v, c\}$ \Comment{$\{u, v, c\}$ violates $(k,2k)$-sparsity}
		\EndIf
		
		\If {$k \geq 3$}
			\State $(X, F_1, \dots, F_{k-1}) \gets $ \Call{BoundedForestDecomp$_{k-1}$}{$D_0[N_G(c)]$}
			\If {$X \neq \varnothing$}
				\State \Return $X \cup \{c\}$ \Comment{$X \cup \{c\}$ violates $(k,2k)$-sparsity}
			\EndIf
			\State $X \gets$ \Call{CheckSparsityWithDecomp$_{k-1,k}$}{$G[N_G(c)], F_1, \dots, F_{k-1}$}
			\If {$X \neq \varnothing$}
				\State \Return $X \cup \{c\}$ \Comment{$X \cup \{c\}$ violates $(k,2k)$-sparsity}
			\EndIf
		\EndIf
		
		\State $T_1, \dots, T_q \gets$ connected components of $(T \setminus \{c\})$
		\For{$i \gets 1, \dots, q$}
			\State $X \gets$ \Call{SaturatedViolation$_{k,2k}$}{$D\left[V(T_i)\right], T_i$}
			\If{$X \neq \varnothing$}
				\State \Return $X$
			\EndIf
		\EndFor
		\State \Return $\varnothing$
    \EndProcedure
  \end{algorithmic}
\end{algorithm}

\begin{proposition}
Algorithm~\ref{alg:csvs_2k} is correct.
For $k\leq2$, it runs in $O(T'_{RTC}(n,k)\log n)$ time, and for $k=3$, it runs in $O(T'_{RTC}(n,k)\log n + n \log^2 n)$ time.
\end{proposition}

\begin{proof}
If \textsc{Reorient}$_k$ returns a set $X$, then $X$ violates $(k,k)$-sparsity.
Since the underlying graph is simple, no set of at most two vertices violates $(k,k)$-sparsity, so $|X|\geq3$ and $X$ also violates $(k,2k)$-sparsity.
Otherwise, $c$ is a source in the returned orientation $D_0$.
By Proposition~\ref{proposition:check_term_superset_sparsity}, the call to \textsc{CheckTermSupsetSparsity} returns exactly when there is a weak violating set $X$ containing $c$ and a vertex outside $N_G(c)\cup\{c\}$; every such set has at least three vertices and therefore violates $(k,2k)$-sparsity.
For $k=1$ and $k=2$, the direct tests cover all remaining violating sets containing $c$.
Suppose now that $k\geq3$.
Since $c$ is a source in $D_0$, every $v\in N_G(c)$ has an incoming arc from $c$, and hence $D_0[N_G(c)]$ is $(k-1)$-indegree-bounded.
If \textsc{BoundedForestDecomp}$_{k-1}$ returns a vertex set $X$, then $X$ violates $(k-1,k-1)$-sparsity and hence also $(k-1,k)$-sparsity.
Since $G[N_G(c)]$ is simple, such a set has size at least three, so Lemma~\ref{lemma:2k_reduction} shows that $X\cup\{c\}$ violates $(k,2k)$-sparsity.
Otherwise, \textsc{BoundedForestDecomp}$_{k-1}$ returns a forest decomposition $F_1,\dots,F_{k-1}$ of $G[N_G(c)]$.
By Proposition~\ref{proposition:middle_l_decomp}, \textsc{CheckSparsityWithDecomp}$_{k-1,k}$ then returns a set $X$ exactly when $G[N_G(c)]$ is not $(k-1,k)$-sparse.
Any such set has size at least three because $G[N_G(c)]$ is simple, so Lemma~\ref{lemma:2k_reduction} shows that $X\cup\{c\}$ is a violating set containing $c$.
Thus the direct tests for $k\leq2$ and the two neighborhood tests for $k\geq3$ cover all remaining violating sets containing $c$.
If a violating set saturated by $T$ avoids $c$, then the observation preceding Algorithm~\ref{alg:csvs} shows that it lies in one component $T_i$ of $T\setminus\{c\}$ and is saturated by $T_i$.
The corresponding recursive call therefore finds a violating set.

For $k\leq2$, the running-time analysis is the same as in Proposition~\ref{proposition:CSVS_time}, with $T'_{RTC}$ in place of $T'_{RC}$.
Now suppose that $k=3$ and consider a recursive call on an $r$-vertex digraph.
The call to \textsc{BoundedForestDecomp}$_2$ takes $O(r\log r)$ time.
If it returns a forest decomposition, Proposition~\ref{proposition:middle_l_decomp} and $T'_{RC}(r,1)=O(r)$ imply that \textsc{CheckSparsityWithDecomp}$_{2,3}$ also takes $O(r\log r)$ time.
Hence the conquer phase of this recursive call takes $O(T_{RTC}(r,3)+r\log r)$ time.
At any fixed recursion level, let $r_1,\dots,r_p$ be the sizes of the recursive instances.
These instances are pairwise vertex-disjoint, so $\sum_{i=1}^p r_i\leq n$.
Using $T_{RTC}(r_i,3)\leq T'_{RTC}(r_i,3)$ and applying Lemma~\ref{lemma:sum_time} to $T'_{RTC}$ gives an $O(T'_{RTC}(n,3))$ bound for the rooted terminal-connectivity terms at that level.
Moreover, $\sum_{i=1}^p r_i\log r_i\leq n\log n$.
Thus each recursion level takes $O(T'_{RTC}(n,3)+n\log n)$ time.
Since the centroid recursion has depth $O(\log n)$, the total running time is $O(T'_{RTC}(n,3)\log n+n\log^2 n)$.
\end{proof}

Finally, by replacing the call to \textsc{SaturatedViolation} with the above subroutine in Algorithm~\ref{alg:middle_l}, we obtain the following algorithm for recognizing $(k, 2k)$-sparse graphs.

\begin{algorithm}[H]
  \caption{Check Sparsity for $\ell=2k$} \label{alg:2k}
  \begin{algorithmic}[1]
    \Input A simple graph $G = (V, E)$.
    \Output A vertex set that violates the $(k, 2k)$-sparsity of $G$ if it exists; otherwise $\varnothing$.
    \vspace{0.2em}
    \hrule
    \vspace{0.2em}
    \Procedure{CheckSparsity$_{k,2k}$}{$G=(V,E)$}
		\State $(X, F_1, \dots, F_k) \gets$ \Call{ForestDecomp$_{k}$}{$G$}
		\If{$X \neq \varnothing$} \Comment{$X$ violates $(k, k)$-sparsity}
			\State \Return $X$
		\EndIf
		\State $D \gets$ a $k$-indegree-bounded orientation of $G$ \Comment{Use $F_1, \dots, F_k$ to obtain $D$} \label{alg:2k:bounded_ori}
		\For{$i \gets 1, \dots, k$}
			\State $T_{i, 1}, \dots, T_{i, q_i} \gets$ connected components of $F_i$
			\For{$j \gets 1, \dots, q_i$}
				\State $X \gets$ \Call{SaturatedViolation$_{k,2k}$}{$D\left[V(T_{i,j})\right]$, $T_{i,j}$}
				\If{$X \neq \varnothing$}
					\State \Return $X$
				\EndIf
			\EndFor
		\EndFor
		\State \Return $\varnothing$
	\EndProcedure
  \end{algorithmic}
\end{algorithm}

\begin{theorem}
Algorithm~\ref{alg:2k} correctly recognizes $(k,2k)$-sparsity on simple graphs.
If the input graph is not $(k,2k)$-sparse, then the algorithm returns a vertex set that violates $(k,2k)$-sparsity; otherwise, it returns $\varnothing$.
Its running time is $O(T_{FD}(n,k)+T'_{RTC}(n,k)\log n)$ for $k\leq2$, and $O(T_{FD}(n, k) + T'_{RTC}(n,k)\log n + n \log^2 n)$ for $k=3$.
\end{theorem}

\begin{proof}
If \textsc{ForestDecomp}$_k$ returns a vertex set $X$, then $X$ violates $(k,k)$-sparsity.
Since the input graph is simple, this implies $|X|\geq3$, and hence $X$ also violates $(k,2k)$-sparsity.
Otherwise, \textsc{ForestDecomp}$_k$ returns a decomposition of $E$ into $k$ forests $F_1,\dots,F_k$.
By Lemma~\ref{lemma:violating_set_2k}, every vertex set that violates $(k,2k)$-sparsity is saturated by at least one of these forests, and hence by the connected component of that forest containing it.
For each such component, Algorithm~\ref{alg:csvs_2k} returns a violating set whenever one saturated by the component exists.
Consequently, Algorithm~\ref{alg:2k} returns a violating set if one exists and returns $\varnothing$ otherwise.

For each fixed $i$, the connected components $T_{i,1},\dots,T_{i,q_i}$ of $F_i$ partition $V$.
Set $n_{i,j}=|V(T_{i,j})|$, so $\sum_{j=1}^{q_i} n_{i,j}=n$.
If $k\leq2$, Lemma~\ref{lemma:sum_time} and $\log n_{i,j}\leq\log n$ show that the total running time of the corresponding \textsc{SaturatedViolation} calls is $O(T'_{RTC}(n,k)\log n)$.
If $k=3$, the same argument bounds the rooted terminal-connectivity terms by $O(T'_{RTC}(n,3)\log n)$, while
\[
\sum_{j=1}^{q_i} n_{i,j}\log^2 n_{i,j}\leq n\log^2 n.
\]
Since $k$ is fixed, summing over the $k$ forests and adding the forest-decomposition cost gives the stated running time.
\end{proof}

\begin{corollary}
Under purely combinatorial implementations, Algorithm~\ref{alg:2k} runs in $O(n\sqrt n)$ time for $k\leq3$.
\end{corollary}

\begin{proof}
By Proposition~\ref{proposition:forest_decomp_combinatorial}, the forest decomposition takes $O(n\sqrt n)$ time.
For $k\leq3$, Table~\ref{tab:parameters} and the definition of $T'_{RTC}$ give $T'_{RTC}(n,k)=O(n)$.
Thus the remaining terms in the preceding theorem are $O(n\log n)$ for $k\leq2$ and $O(n\log^2 n)$ for $k=3$.
Since $n\log^2 n=O(n\sqrt n)$, the claim follows.
\end{proof}

For $k>3$, the correctness proof of Algorithm~\ref{alg:2k} remains unchanged, and a refined analysis of the centroid recursion shows that its running time is $O(n^2)$.
In practice, however, the more general algorithm presented below is preferable, since for $\ell=2k$ it achieves the same $O(n^2)$ running time and is considerably simpler.

\subsubsection{A general algorithm}
The general recognition algorithm for the range $2k \leq \ell < 3k$ follows the usual pebble game paradigm. Given a graph $G = (V, E)$, we construct a $(k, \ell)$-sparse subgraph $H$ incrementally, starting from the empty graph on vertex set $V$ and considering the edges of $G$ one by one in an arbitrary but fixed order. For each edge $uv \in E$, we test whether adding $uv$ to the current subgraph $H$ preserves $(k, \ell)$-sparsity. If so, we insert $uv$ into $H$; otherwise, $G$ is not $(k, \ell)$-sparse, and the algorithm terminates. The next observation gives an exact characterization of when such an insertion is possible.

\begin{observation}
\label{observation:big_l}
Let $2k \leq \ell < 3k$, let $H = (V,E)$ be a $(k,\ell)$-sparse graph, and let $u,v \in V$ be distinct vertices with $uv \notin E$. Then the graph $H' = (V, E \cup \{uv\})$ is $(k,\ell)$-sparse if and only if $H$ contains no vertex set $X \supsetneq \{u,v\}$ that violates $(k,\ell+1)$-sparsity in $H$.
\end{observation}

\begin{proof}
Since $H$ is $(k,\ell)$-sparse, every vertex set $X$ that does not contain both $u$ and $v$ continues to satisfy the sparsity bound after the edge $uv$ is added. Therefore, if $H'$ is not $(k,\ell)$-sparse, then there must exist a violating vertex set $X$ that contains both $u$ and $v$. As the sparsity condition is imposed only on vertex sets of size at least three, such a set must satisfy $X \supsetneq \{u,v\}$. Now let $X \supsetneq \{u,v\}$ be any vertex set. Since both $u$ and $v$ lie in $X$, we have
\[
i_{H'}(X)=i_H(X)+1.
\]
It follows that $X$ violates $(k,\ell)$-sparsity in $H'$ if and only if it violates $(k,\ell+1)$-sparsity in $H$.
\end{proof}

By the observation above and the case $t=2$ of Lemma~\ref{lemma:main}, deciding whether $H'$ is $(k,\ell)$-sparse reduces to a rooted arc-connectivity problem, provided that a $k$-indegree-bounded orientation of $H$ is already available. As shown in the implementation below, such an orientation can be maintained efficiently throughout the algorithm.

\begin{algorithm}[H]
  \caption{Check Sparsity for $2k \leq \ell < 3k$} \label{alg:big_l}
  \begin{algorithmic}[1]
    \Input A simple graph $G = (V, E)$.
    \Output A vertex set that violates the $(k, \ell)$-sparsity of $G$ if it exists; otherwise $\varnothing$.
    \vspace{0.2em}
    \hrule
    \vspace{0.2em}
    \Procedure{CheckSparsity$_{k,\ell}$}{$G=(V,E)$}
		\State $H \gets (V, \emptyset)$ \Comment{Initialize our sparse subgraph}
		\State $D \gets (V, \emptyset)$ \Comment{Initialize our $k$-indegree-bounded orientation of $H$}
		\For{$uv \in E$}
    		\State $(X, D) \gets$ \Call{Reorient$_{k}$}{$D$, $\{u, v\}$} \label{alg:big_l:d0} \Comment{It is guaranteed that $X = \varnothing$ and $D \neq \varnothing$}
    		\State $X \gets$ \Call{CheckSupsetSparsity$_{k,\ell+1}$}{$D$, $\{u, v\}$} \label{alg:big_l:checksupset}
			\If{$X \neq \varnothing$} \Comment{$X$ violates $(k,\ell)$-sparsity in $H \cup \{uv\}$} \label{alg:big_l:check}
				\State \Return $X$
			\EndIf
			\State $H \gets H \cup \{uv\}$ \Comment{Insert the edge $uv$ into $H$}
			\State $D \gets D \cup \{uv\}$ \Comment{Insert an arc from $u$ to $v$ in $D$} \label{alg:big_l:d_augment}
		\EndFor
		\State \Return $\varnothing$
	\EndProcedure
  \end{algorithmic}
\end{algorithm}

\begin{note}
By Lemma~\ref{lemma:u0_source}, the subroutine \textsc{Reorient} always finds a suitable orientation, since the set $\{u, v\}$ is independent in $H$, and at the beginning of each iteration of the \texttt{for} loop the graph $H$ is $(k,2k)$-sparse. After the arc $uv$ is inserted at line~\ref{alg:big_l:d_augment}, the orientation $D$ remains $k$-indegree-bounded. Indeed, the only indegree that increases is that of $v$, and before the insertion we have $\varrho_D(v)=0<k$.
\end{note}

\begin{theorem}
Algorithm~\ref{alg:big_l} correctly recognizes $(k,\ell)$-sparsity for $2k \leq \ell < 3k$ on simple graphs.
If the input graph is not $(k,\ell)$-sparse, then the algorithm returns a vertex set that violates $(k,\ell)$-sparsity; otherwise, it returns~$\varnothing$.
Its running time is $O(n \cdot T_{RC}(n,\ell + 1 - 2k))$.
\end{theorem}

\begin{proof}
We maintain the invariant that $H$ is $(k,\ell)$-sparse and $D$ is a $k$-indegree-bounded orientation of $H$.
The invariant holds initially.
At the beginning of an iteration for an edge $uv$, the note above guarantees that \textsc{Reorient}$_k$ produces a $k$-indegree-bounded $\{u,v\}$-source orientation.
By Observation~\ref{observation:big_l} and Lemma~\ref{lemma:main} with $t=2$ and sparsity parameter $\ell+1$, the subsequent call to \textsc{CheckSupsetSparsity} returns a vertex set precisely when $H\cup\{uv\}$ is not $(k,\ell)$-sparse.
Any returned set violates $(k,\ell)$-sparsity in $H\cup\{uv\}$ and therefore also in the input graph $G$.
If no such set is returned, then inserting $uv$ preserves the sparsity of $H$, and the note above shows that the updated orientation $D$ remains $k$-indegree-bounded.
Thus, if the algorithm reaches the end, then $H=G$ is $(k,\ell)$-sparse, which proves correctness.

By Propositions~\ref{proposition:main} and~\ref{proposition:d0_orientation}, the calls to the subroutines \textsc{CheckSupsetSparsity} and \textsc{Reorient} run in $O(T_{RC}(n,\ell + 1 - 2k))$ and $O(n)$ time, respectively. Since the \texttt{for} loop has $m = O(n)$ iterations, the overall running time~is
\[
O(n \cdot (T_{RC}(n,\ell + 1 - 2k) + n))
= O(n \cdot T_{RC}(n,\ell + 1 - 2k)).
\qedhere
\]
\end{proof}

By Table~\ref{tab:parameters}, this bound is $O(n^2)$ when $\ell\leq 2k+2$ and $O(n^2\log n)$ otherwise.

\begin{note}
It is also possible to decide the sparsity of $H'$ without invoking Lemma~\ref{lemma:main} by extending the ideas developed in~\cite{madarasi2023klSparse}. More precisely, to handle an edge $uv$, one may check, for each vertex $w \in V \setminus \{u,v\}$, whether there exists a $k$-indegree-bounded orientation of $H$ in which the sum of the indegrees of $u$, $v$, and $w$ is at most $3k-\ell-1$. As in the pebble game algorithm, each such test can be implemented using augmenting paths, and the entire insertion step can be carried out in quadratic time.
This yields a conceptually simpler sparsity test than Algorithm~\ref{alg:big_l}. However, its overall running time is $O(n^3)$ and is therefore asymptotically~worse.
\end{note}

\begin{note}
With a minor modification, Algorithm~\ref{alg:big_l} can be extended to compute an inclusion-wise maximal $(k,\ell)$-sparse subgraph of $G$. For a graph with $n$ vertices and $m$ edges, the extended algorithm runs in~time
\[
  O(m \cdot T_{RC}(n,\ell + 1 - 2k)).
\]
\end{note}

\section*{Acknowledgment}
The work was supported by the Lend\"ulet Programme of the Hungarian Academy of Sciences -- grant number LP2021-1/2021, and by the Ministry of Innovation and Technology of Hungary from the National Research, Development and Innovation Fund, financed under the ELTE TKP 2021-NKTA-62 funding scheme.

\bibliographystyle{plain}
\bibliography{bibliography}

\end{document}